\documentclass[12pt]{article}

\usepackage{amsmath}
\usepackage{amssymb}
\usepackage{amsthm}
\usepackage{amsfonts}
\usepackage{latexsym}
\usepackage{cite}
\usepackage{epsfig}
\usepackage{graphicx}
\usepackage{mathrsfs}
\usepackage{url}
\usepackage{color}
\usepackage{eufrak}
\usepackage{multirow}
\usepackage{hyperref}

\usepackage{scalerel}
\usepackage{tikz}
\usetikzlibrary{svg.path}

\definecolor{orcidlogocol}{HTML}{A6CE39}
\tikzset{
  orcidlogo/.pic={
    \fill[orcidlogocol] svg{M256,128c0,70.7-57.3,128-128,128C57.3,256,0,198.7,0,128C0,57.3,57.3,0,128,0C198.7,0,256,57.3,256,128z};
    \fill[white] svg{M86.3,186.2H70.9V79.1h15.4v48.4V186.2z}
                 svg{M108.9,79.1h41.6c39.6,0,57,28.3,57,53.6c0,27.5-21.5,53.6-56.8,53.6h-41.8V79.1z M124.3,172.4h24.5c34.9,0,42.9-26.5,42.9-39.7c0-21.5-13.7-39.7-43.7-39.7h-23.7V172.4z}
                 svg{M88.7,56.8c0,5.5-4.5,10.1-10.1,10.1c-5.6,0-10.1-4.6-10.1-10.1c0-5.6,4.5-10.1,10.1-10.1C84.2,46.7,88.7,51.3,88.7,56.8z};
  }
}

\newcommand\orcidicon[1]{\href{https://orcid.org/#1}{\mbox{\scalerel*{
\begin{tikzpicture}[yscale=-1,transform shape]
\pic{orcidlogo};
\end{tikzpicture}
}{|}}}}

\usepackage{hyperref} 

\usepackage{longtable}

\makeatletter
\@addtoreset{equation}{section}
\makeatother

\makeatletter
\@addtoreset{table}{section}
\makeatother

\newtheorem{theorem}{Theorem}[section]
\newtheorem{lemma}[theorem]{Lemma}

\theoremstyle{definition}
\newtheorem{definition}[theorem]{Definition}

\newtheorem{remark}[theorem]{Remark}

\newcommand{\Vc}{\mathcal{V}}

\newcommand{\PG}{\mathrm{PG}}
\newcommand{\QM}{\overline{\mathrm{QM}}}

\newcommand{\zb}{\mathbf{0}}

\newcommand{\F}{\mathbb{F}}

\newcommand{\Bs}{\mathscr{B}}
\newcommand{\Ps}{\mathscr{P}}

\newcommand{\Ab}{\boldsymbol{A}}

\newcommand{\D}{\boldsymbol{D}}
\newcommand{\Hb}{\boldsymbol{H}}
\newcommand{\hb}{\boldsymbol{h}}
\newcommand{\Ib}{\boldsymbol{I}}
\newcommand{\Mb}{\boldsymbol{M}}
\newcommand{\Ub}{\boldsymbol{U}}
\newcommand{\ub}{\boldsymbol{u}}
\newcommand{\pib}{\boldsymbol{\pi}}
\newcommand{\Wb}{\boldsymbol{W}}
\newcommand{\Hcb}{\boldsymbol{\mathcal{H}}}

\newcommand{\wb}{\boldsymbol{w}}

\newcommand{\xb}{\boldsymbol{x}}
\newcommand{\yb}{\boldsymbol{y}}
\newcommand{\zbl}{\boldsymbol{z}}

\newcommand{\ov}{\overline{v}}
\newcommand{\oc}{\overline{c}}

\newcommand{\ozero}{\overline{0}}

\newcommand{\T}{\text}
\newcommand{\db}{\displaybreak[3]}

\begin{document}

\title{
Further results on binary linear covering codes and saturating sets in projective spaces}
\date{}
\maketitle
\begin{center}
{\sc Alexander A. Davydov \orcidicon{0000-0002-5827-4560}}\\
 {\sc\small Kharkevich Institute for Information Transmission Problems}\\
 {\sc\small Russian Academy of Sciences,
Moscow, 127051, Russian Federation}\\
 \emph{E-mail address:} alexander.davydov121@gmail.com\medskip\\
 {\sc Stefano Marcugini  \orcidicon{0000-0002-7961-0260},
 Fernanda Pambianco  \orcidicon{0000-0001-5476-5365}}\\
 {\sc\small Department of  Mathematics  and Computer Science, University of Perugia,}\\
 {\sc\small Perugia, 06123, Italy}\\
 \emph{E-mail addresses:} \{stefano.marcugini, fernanda.pambianco\}@unipg.it\medskip\\
 {\sc Stephen Wu \orcidicon{0009-0003-8471-2397}}\\
 {\sc\small Independent, San Francisco, California, USA}\\
  \emph{E-mail address:} wujstephen@gmail.com
\end{center}

\textbf{Abstract.}
The binary length function $\ell_2(r,R)$ is the smallest
length of a binary linear code with codimension (redundancy) $r$ and covering radius $R$. Let $s_2(N,\rho)$ be the smallest size of a $\rho$-saturating set in the binary projective space $\mathrm{PG}(N,2)$. It is known that the binary length function satisfies $\ell_2(r,R)=s_2(r-1,R-1)$. We obtain the following new upper bounds on $\ell_2(r,R)$, improving the previously known upper bounds by $\Delta(r,R)$: $R=2,r=2t\ge28$, and $r=10,18,20$, $\ell_2(r,2)=s_2(r-1,1)\le51\cdot2^{r/2-5}-1,\Delta(r,2)=2^{r/2-5}$; $R=3,r=3t-1\ge41$ and $r=26$, $\ell_2(r,3)=s_2(r-1,2)\le818\cdot2^{(r-26)/3}-1$, $\Delta(41,3)=63,\Delta(r,3)=2^{(r-26)/3}$ if $r\ge44$. To obtain these bounds, we construct new infinite code families, using distinct versions of the $q^m$-concatenating constructions of covering codes, including new versions proposed in this paper. We also obtain new useful partitions of column sets of parity check matrices of some codes. The new code families provide the asymptotic covering-density bounds $\overline{\mu}(2)\le2601/2048<1.27002$, $\overline{\mu}(3)\le1.35935$, improving the previously known bounds and giving rise to the new upper bounds $f(2)\le2601/2048<1.27002$, $f(3)\le1.35935$ on the constants $f(2),f(3)$ of the Green's Open Problem 40.

\textbf{Keywords:} Binary covering codes, the length function, covering radius

\textbf{Mathematics Subject Classification (2010).} 94B65, 94B25, 94B60, 94B05

\section{Introduction, notations, and background}\label{sec1:Intro}
Let $\F_{q}$ be the \emph{Galois field} with $q$ elements, $\F_q^*=\F_{q}\setminus\{0\}$. Throughout the paper we use the field $\F_2$. Let $\triangleq$ denote the sign ``equality by definition''. The \emph{$n$-dimensional vector space over $\F_2$} is denoted by $\F_2^n$. Let $\F_2^{m\times n}$ be the set of binary $m \times n$ matrices. Let $\#S$ be the cardinality of a set $S$.  Let a $t$-set (resp. a $t$-subset) be a set (resp. a subset) of cardinality~$t$. We use lowercase letters, e.g. $v$, for scalars, and overlined lowercase letters, e.g. $\ov$, for vectors. Let $\ozero$ be the zero vector the length of which is clear by the context. Three distinct nonzero binary vectors are linearly dependent if and only if their sum is the zero vector. Let $tr$ denote transposition.
We denote by $\Ib_{r}$ the identity matrix of order $r$, whose first column is $(10\ldots0)^{tr}$.

The \emph{Hamming weight} $\T{wt}(\ov)$ of a vector $\ov\in\F_2^n$ is the number of its non-zero coordinates.
The \emph{Hamming distance} $d(\ov,\ov')$ between two vectors $\ov$, $\ov'$ is the number of coordinates in which they differ.
The \emph{binary Hamming ball} of radius $\varrho$ with center $\ov\in \F_2^n$ is the set $\{\ov'\,|\,\ov'\in \F_2^n,~d(\ov,\ov')\leq \varrho\}$, whose volume is $V_{\varrho,n,2}=\sum_{i=0}^\varrho\binom{n}{i}$.

A \emph{binary linear code} $C$ of length $n$ and dimension $k$ is  a $k$-dimensional subspace of $\F_2^n$; it is denoted as $[n,k]_2$ code, $\#C=2^k$. Let $r\triangleq n-k$ be the \emph{codimension} (or \emph{redundancy}) of an $[n,k]_2$ code.
A vector $\oc\in C$ is a \emph{codeword} of $C$. The \emph{minimum distance} (or simply \emph{distance}) $d$ of a code $C$ is the minimum Hamming distance between two distinct codewords.
A binary $(n-k)\times n$ matrix $\Hb\in\F_2^{(n-k)\times n}$ is a \emph{parity check matrix} of a binary linear $[n,k]_2$ code $C$ if $\oc\Hb^{tr}=\ozero$ for every codeword $\oc\in C$ and rank $\Hb=n-k$.

\begin{definition} \label{Def1:CoverRad}
A binary linear $[n,k]_2$ code $C$ has \emph{covering radius} $R$ if any of the following equivalent conditions  holds:

\textbf{(i)} The value $R$  is the smallest integer such that the space $\F_2^n$ is covered by binary Hamming balls of radius $R$ centered at all the codewords of $C$.

\textbf{(ii)}
Every vector in $\F_2^{n-k}$ can be expressed as a sum of at most $R$ columns of an $(n-k)\times n$ parity check matrix of $C$, and $R$ is the smallest value with this property.
\end{definition}

An $[n,k]_2$ code with minimum distance $d$ and covering radius $R$ is denoted $[n,k,d]_2R$,
where $d$ and $R$ can be omitted when they are not relevant, not known, or clear from the context.

\begin{definition}\label{def1:LO}
\cite[Definition 1]{DFMP-LO-CovCodIEEE2005}
A linear covering code is called \emph{locally optimal (LO)}
if one cannot remove any column from the parity check matrix of the
code without increasing its covering radius. An LO code is also called
 non-shortening in the sense mentioned.
\end{definition}

For an introduction to coding theory, see \cite{HufPless,MWS,Roth} and the references therein.

\begin{definition} \label{Def1:CoverDens}
The \emph{covering density} $\mu(n,R,C)$  of an $[n,k]_2R$-code $C$ is the ratio of the total volume of the $2^k$ Hamming balls of radius $R$, centered at the codewords of $C$, to the total volume $2^n$ of the ambient space $\F_2^n$. We have
\begin{equation}
\mu(n,R,C)\triangleq\frac{2^k V_{R,n,2}}{2^n}=\frac{1}{2^{n-k}}\sum_{i=0}^R\binom{n}{i}\ge1.\label{eq1:CovDensity}
\end{equation}
We may write $\mu(C)$, $\mu(n-k)$, $\mu(r)$, $\mu(r,R)$, $\mu(n,R)$, or $\mu$ if this is clear from the context.
\end{definition}

  Let $U$ be an infinite family of codes with covering radius $R$, and let $U_n$ be a code in $U$ with length $n$. For the family $U$ we consider the asymptotic parameter \cite{Dav90PIT,GabDavTombR=2}:
\begin{equation}\label{eq1:liminfdens}
  \overline{\mu}(R,U)\triangleq\liminf_{n\rightarrow\infty,\,U_n\in U}\mu(n,R,U_n).
\end{equation}
We may write $\overline{\mu}(R)$  when $U$ is defined by the context.

The covering quality of a code is better if its covering density is smaller.
For fixed codimension $n-k$ and covering radius $R$, the covering density $\mu(n,R,C)$ decreases when $n$ decreases.

\begin{definition}\label{def1:length function}
The \emph{length function} $\ell_2(r,R)$ is the smallest
length of a binary linear code with codimension (redundancy) $r$ and covering radius $R$.
\end{definition}

The study of covering codes is a classical combinatorial problem. Covering codes are connected to many areas of information theory, combinatorics, and practical applications, see \cite[Section 1.2]{CHLL-bookCovCod}, \cite[Section 1.2]{DMP-arXiv2025}, \cite[Section IB]{DMP-IEEE2026}, and the references therein.

For an introduction to coverings of vector Hamming spaces over finite fields, see
\cite{GrahSlo,Handbook-coverings,CHLL-bookCovCod,DGMP-AMC,HufPless,Delsarte} and the references therein. The results and problems, related to covering codes, are discussed in numerous works; for example, see the online bibliography \cite{LobstBibl}, the papers
\cite{BaichVavLengthFunBin,Graismer-2024,BrPlWi,Dav90PIT,DDL-ACCT2-1990,DDL-ACCT3-1992,DDL-IEEE,DFMP-LO-CovCodIEEE2005,%
DMP-arXiv2025,DMP-IEEE2026,DMPW-arXiv2026,DavOst-IEEE2001,DavOst-DESI2010,HonkLits1996,Giul2013Survey,Janwa,KabPan,KaikRoseADSlike2003,%
OstKaikUpBndBin1998,Struik,PlanFuncCovCodWuPan2025,DougJanCovRadCalc1991,Wu2026} and \cite[Section 5, Problem 40]{Green}.

Let $\PG(N,2)$ be the binary $N$-dimensional projective space over $\F_2$.
  A point set $S\subseteq\PG(N,2)$ is called $\rho$-\emph{saturating} if every point of $\PG(N,2)$ can be represented as a sum of at most $\rho+1$ points of $S$ and $\rho$ is the smallest value with this property. Saturating sets are also called ``saturated sets'', ``spanning sets'', and ``dense sets''. A $\rho$-\emph{saturating} set is called \emph{minimal} if one cannot remove any point from the set without losing the $\rho$-\emph{saturating} property.

Let $s_2(N,\rho)$ be the smallest size of a $\rho$-saturating set in $\PG(N,2)$.
Let the columns of a parity check matrix $\Hb$ of an $[n,n-r]_2R$ code be nonzero and pairwise distinct; if the positions of a column of $\Hb$ are interpreted as coordinates of a point in $\PG(r-1,2)$, then this matrix corresponds to an $(R-1)$-saturating $n$-set in $\PG(r-1,2)$, and conversely. Thus, there is a \emph{one-to-one correspondence} between $[n,n-r]_2R$ codes and $(R-1)$-saturating $n$-sets in $\PG(r-1,2)$. This implies
\begin{equation}\label{eq1:cov=sat}
  \ell_2(r,R)=s_2(r-1,R-1).
\end{equation}
A minimal $\rho$-saturating set corresponds to a locally optimal covering code.

For an introduction to geometries over finite fields and their connections with coding theory, see \cite{Dav95,DGMP-AMC,EtzStorm2016,Giul2013Survey,Hirs,HirsStor-2001} and the references therein.

\textbf{Green's Open Problem 40}. \cite[Section 5, Problem 40]{Green} Let $R$ be a fixed positive integer. Let $f(R)$ be the smallest constant with the following property: there exists an infinite sequence of integers $n$ together with subspaces $V_n\subseteq\F_2^n$ such that the space $\F_2^n$ is covered by binary Hamming balls of radius $R$, centered at all the vectors of $V_n$, and $\#V_n=(f(R) +o(1))2^n/V_{R,n,2}$. Does $f(R)\rightarrow\infty$?

By \eqref{eq1:CovDensity} and \eqref{eq1:liminfdens}, the Green's Open Problem 40 is connected with finding infinite families of binary codes with relatively small asymptotic covering density; an infinite code family providing $\overline{\mu}(R)\le c$, where $c$ is a constant, gives rise to estimate $f(R)\le c$.

In \cite{Green}, it is noted that $f(1)=1$ (due to the Hamming codes) and the best-known upper bound on $f(2)$ (at the time of writing \cite{Green} 2013-2018) is $f(2)\le 1.4238$ (with reference to \cite{Dav90PIT}). We note a minor inaccuracy in the reference: in fact, the value 1.4238 was represented in the conference ACCT-3 in 1992, see \cite[p.\ 53]{DDL-ACCT3-1992}, and was subsequently published in \cite[Example 3.1, equation (1.3)]{DDL-IEEE} in 1994. At the same time, both the papers \cite{DDL-ACCT3-1992,DDL-IEEE} use the methods proposed in \cite{Dav90PIT}.

In \cite[equation (2.1), Theorem 5.7]{DMP-arXiv2025}, as of November 2025, this upper bound had been improved to $f(2)\le 1.32031$. Furthermore, in \cite[equation (2.2), Theorem 7.7]{DMP-arXiv2025}, the bound $f(3)\le1.36433$ was obtained.
The results of \cite{DMP-arXiv2025} were subsequently published in \cite{DMP-IEEE2026}. Note also that in \cite[Example 5.5, equation (1.5)]{DDL-IEEE}, an infinite family of codes of covering radius $R=4$ was obtained with asymptotic covering density $\overline{\mu}(4)\le 2.3392$ that gives rise to upper bound $f(4)\le2.3392$. As far as the authors know, this estimate is not improved yet.

\emph{This paper} is devoted to \emph{new constructive upper bounds on the binary length functions} $\ell_2(r,R)$.
If there is an $[n,n-r]_2R$ code, then $\ell_2(r,R)\le n$. For $q=2$, $R=2,3$, we consider the known versions of the $q^m$-concatenating constructions of covering codes, see e.g. \cite{Dav90PIT,DDL-IEEE,DGMP-AMC,DMP-arXiv2025,DMP-IEEE2026,DavOst-IEEE2001,DavOst-DESI2010}, and propose new ones. To use the constructions effectively, we obtain new useful partitions of column sets of parity check matrices for certain covering codes. As a result, we obtain new infinite families of covering codes that provide improved upper bounds on $\ell_2(r,R)$ and smaller asymptotic covering densities $\overline{\mu}(R)$ \eqref{eq1:liminfdens} compared with previously  known values; this implies the improvement of the upper bounds on the constant of the Green's Open Problem 40 to $f(2)\le2601/2048<1.27002$, $f(3)\le1.35935$. Also, we obtain a few sporadic new covering codes, see Section~\ref{sec2:main res} for details.

The new $[50,40,3]_22$ code and the partition of column set of its parity check matrix into 10 subsets, proposed by Stephen Wu, see \cite{Wu2026}, play a very important role in obtaining new results. In particular, the code is used as a seed in the $q^m$-concatenating constructions for $q=2$, $R=2$ and as an auxiliary element of the construction for $q=2$, $R=3$.

The paper is organized as follows. In Section \ref{sec2:main res}, we collect the new results  and compare them with the known ones.
In Section \ref{sec3:qm concat}, we give a description of a basic $q^m$-concatenating construction for $q=2$. In Sections \ref{sec4:R=2known}, \ref{sec6:R=3known}, we describe the known results, useful for this paper. In Sections \ref{sec5:R=2new}, \ref{sec7:R=3new}, we obtain the new results.

\section{The main results}\label{sec2:main res}
Let the asymptotic covering density $\overline{\mu}(R)$ be as in \eqref{eq1:liminfdens}. Let $f(R)$ be the constant from Green's Open Problem 40 \cite[Section 5, Problem 40]{Green}. Let $t$ be an integer. Throughout the paper, let $\Delta(r,R)$ denote the decrease of the known upper bounds on $\ell_2(r,R)$ provided by the new results.

For $R=2,r=2t$, and $R=3,r=3t-1$, the best infinite families of binary covering $[n,n-r]_2R$ codes, in terms of the asymptotic density $\overline{\mu}(R)$, as the authors are aware, are obtained in \cite{DMP-arXiv2025,DMP-IEEE2026} with the following parameters:
\begin{align*}
&R=2,\,r=2t,\,r=10,18,20,\,\T{and}\,r\ge28,\,n=26\cdot2^{r/2-4}-1,\,\overline{\mu}(2)\le1.32031 \db\\
&\T{\cite[equation (2.1), Theorem 5.7]{DMP-arXiv2025},\,\cite[equation (4), Theorem 11]{DMP-IEEE2026}}.\db\\
&R=3,~r=3t-1,~r=26\T{ and }r\ge44,~n=819\cdot2^{(r-26)/3}-1,~\overline{\mu}(3)\le1.36433\db\\
&\T{\cite[equation (2.2), Theorem 7.7]{DMP-arXiv2025},\,\cite[equation (5), Theorem 19]{DMP-IEEE2026}}.
\end{align*}
This implies the following estimate, see \eqref{eq1:cov=sat}:
\begin{align*}
&R=2,\,r=2t,\,\ell_2(r,2)=s_2(r-1,1)\le26\cdot2^{r/2-4}-1,~f(2)\le1.32031;\db\\
&R=3,~r=3t-1,~\ell_2(r,3)=s_2(r-1,2)\le819\cdot2^{(r-26)/3}-1,~f(3)\le1.36433.
\end{align*}

In this paper, these known results are improved. In Sections \ref{sec5:R=2new}, \ref{sec7:R=3new},  Theorems \ref{th5:Wu code}, \ref{th5:r=18}, \ref{th5:r=28}, \ref{th5:R=2NewInfFam}, \ref{th7:r=26,38}, \ref{th7:newr=3t-1a},
 we obtain new infinite families of binary covering codes, whose parameters are collected in Theorem \ref{th2:InfFamNew}.
Also, in Theorems \ref{th5:-3r=22,24,26}, \ref{th7:151R=3}, \ref{th7:r=26,38}, we obtain  sporadic new covering codes, whose parameters are collected in  Theorem \ref{th2:SporasNew}. The new code families and sporadic new codes provide the decrease of the known upper bounds on $\ell_2(r,R)$, see the values $\Delta(r,R)$ in Theorems \ref{th2:InfFamNew}, \ref{th2:SporasNew}.

\begin{theorem}\label{th2:InfFamNew}
 There exist new infinite families of covering $[n,n-r]_2R$ codes such that:
\begin{align}\label{eq2:InfFamR=2New}
&R=2,~r=2t\ge28\T{ and }r=10,18,20,~t=5,9,10,\T{ and }t\ge14, \db\\
&n=51\cdot2^{r/2-5}-1= 51\cdot2^{t-5}-1,~\overline{\mu}(2)\le\frac{2601}{2048}<1.27002;\notag\db\\
&\ell_2(r,2)=s_2(r-1,1)\le51\cdot2^{r/2-5}-1,~\Delta(r,2)=2^{r/2-5}. \notag\db\\
&R=3,~r=3t-1\ge41\T{ and }r=26,~t=9\T{ and }t\ge14,~\ell=2, \label{eq2:InfFamR=3New}\db\\
&n=818\cdot2^{(r-26)/3}-1=818\cdot2^{t-9}-1,~\overline{\mu}(3)\le1.35935;\notag\db\\
&\ell_2(r,3)=s_2(r-1,2)\le818\cdot2^{t-9}-1,~\Delta(41,3)=63,~\Delta(r,3)=2^{(r-26)/3}\T{ if }r\ge44.
\notag
\end{align}
These infinite families provide the upper bounds $f(2)\le2601/2048<1.27002$, $f(3)\le1.35935$, on the constants $f(R)$ of the Green's Open Problem 40 \cite[Section 5, Problem 40]{Green}.
\end{theorem}

\begin{theorem}\label{th2:SporasNew}
There exist new covering $[n,n-r]_2R$ codes such that:
\begin{align*}
&R=2,~r=22,24,26,~n=52\cdot2^{r/2-5}-3;~\ell_2(r,2)=s_2(r-1,1)\le n,~\Delta(r,2)=2^{r/2-5}.\db\\
&R=3,\,r=38,\,n=819\cdot2^{(r-26)/3}-2,\,\ell_2(r,3)=s_2(r-1,2)\le n;\,\Delta(r,3)=2^{(r-26)/3}.\db\\
&R=3,~r=18,~n=151,~\ell_2(r,3)=s_2(r-1,2)\le n;~ \Delta(18,3)=2.
\end{align*}
\end{theorem}

\section{Basic $\boldsymbol{q^m}$-concatenating Construction QM for $\boldsymbol{q=2}$}\label{sec3:qm concat}

Throughout this paper, all matrices and columns are binary. An element of $\F_{2^m}$ written in a binary matrix denotes an $m$-dimensional binary column, that is a binary representation of this element; and vice versa, an $m$-dimensional binary column can be
viewed as an element of $\F_{2^m}$.

\begin{definition} \label{def31:partition} \cite[Definition 2.1]{DDL-IEEE}, \cite[Definition 2.1]{DGMP-AMC}
Let $\Hb\in\F_2^{\,r\times n}$ be a parity check $r\times n$ matrix
of an $ [n,n-r]_2R$ code $C$ and let $0\le\ell\le R$.
\begin{description}
  \item[(i)] A partition of the column set of the matrix
$\Hb$ into nonempty subsets is called an  $(R,\ell)$\emph{-partition} if every vector of $\F_2^{\,r}$, including
the zero vector, can be represented as the sum of at least $\ell $ and at most $R$ distinct columns of
$\Hb$ belonging to distinct subsets. For an $(R,0)$-partition we can formally treat the zero column as the sum of 0 columns. We use the term ``$R$-\emph{partition}'' when the value of $\ell$ is not relevant or not known.

  \item[(ii)] If $\Hb$ admits an $(R,\ell)$-partition, the code $C$ is called an $(R,\ell) $\emph{-object} of the space $\F_2^{\,n}$ and is denoted as an $[n,n-r]_2R,\ell $ or an $[n,n-r,d]_2R,\ell$ code. It is not necessary that $\ell $ is the greatest value with the properties considered. Therefore, the same $(R,\ell )$-partition with $\ell\ge1$ is also an $(R,\ell_1)$-partition with $\ell_1=0,1,\ldots,\ell -1$.  So, an $[n,n-r]_2R,\ell$ code with $\ell\ge1$ is also an $[n,n-r]_2R,\ell_1$ code with $\ell_1 = 0, 1,\ldots,\ell - 1$.
\end{description}
\end{definition}

 As will be shown later in the paper, using $(R,\ell)$-objects with $\ell\ge1$ allows us to improve parameters of new codes obtained by the $q^m$-concatenating constructions.

Let $p(\Hb,\ell;\Ps)$ be the number of subsets in an $(R,\ell)$-partition $\Ps$ of a parity check matrix $\Hb$ of an $[n,n-r]_2R,\ell$ code. Clearly, $R \le p(\Hb,\ell; \Ps)\le n$. If $p(\Hb,\ell;\Ps)=n$, i.e. every subset contains one column, then the partition $\Ps$ is called \emph{trivial}. The trivial partition of a parity check matrix of an $ [n,n-r]_2R,\ell $ code is an $(R,\ell )$-partition. We use the notation $p(\Hb,\ell)$ when the partition $\Ps$ is not relevant. In other words, the notation $p(\Hb,\ell)$ means that the matrix  $\Hb$ admits an $(R,\ell)$-partition into $p(\Hb,\ell)$ subsets. For the notations $p(\Hb,\ell;\Ps)$ and $p(\Hb,\ell)$, the value of $R$  is defined by the context.

\begin{lemma}\label{Lem31:d&elll}
\emph{\cite{Dav90PIT,Dav95,DDL-IEEE}, \cite[Lemma 2.2]{DGMP-AMC}} An $[n,n-r,d]_2R$ code admits an $(R,\ell)$-partition
with $\ell \geq 1$ if and only if $d\leq
R. $ If $d>R$ the maximum possible value of $\ell $ is $0$.
\end{lemma}

Throughout the paper, all examined codes have minimum distance $d\ge3$. So, all codes of covering radius $R=2$ are $(2,0)$-objects.
Therefore, for codes of covering radius $R=2$, we use the shorthand notations $[n,n-r]_22$ and $[n,n-r,d]_22$ codes, 2-objects, 2-partition, $p(\Hb;\Ps)$, $p(\Hb)$.

The $q^m$-concatenating constructions were proposed in 1990 in \cite{Dav90PIT} and then investigated and developed in numerous works, see e.g. \cite{Handbook-coverings,Dav95,DDL-ACCT2-1990,DDL-ACCT3-1992,DDL-IEEE,DFMP-LO-CovCodIEEE2005,DGMP-AMC,DMP-arXiv2025,DMP-IEEE2026,Giul2013Survey}, \cite[Section 5.4]{CHLL-bookCovCod}, \cite[Supplement]{Struik}, and the references therein. The term ``$q^m$-concatenating constructions'' is  not always used in these papers. Here we call the basic construction QM because the notation resembles $q^m$. The name QM is used in the literature, see e.g. \cite[Section 2]{DGMP-AMC}, \cite{DMP-arXiv2025,DMP-IEEE2026}.

Starting with an $[n_0,n_0-r_0]_qR$ code of length $n_0$, the $q^m$-concatenating constructions yield  new $[n,n-(r_0+Rm)]_qR$ codes with the same covering radius $R$ and length $ n=q^m n_0+N_m$, where $m$ is a quantity satisfying the conditions specified below and $N_m\le R\theta _{m,q}$, with $\theta_{m,q}=(q^m-1)/(q-1)$, $\theta_{m,2}=2^m-1$.
The covering density of the new codes is almost the same (slightly greater) as for the starting code. Using the constructions iteratively, we can obtain infinite families of new $ [n,n-(r_0+Rm)]_qR$ codes where $m$ ranges over an infinite set of integers.

\textbf{Basic Construction QM for $\boldsymbol{q=2}$.}
We use the results and ideas of \cite{Dav90PIT,Dav95,DDL-ACCT3-1992,DDL-IEEE,DGMP-AMC,DMP-arXiv2025,DMP-IEEE2026}.

Let $\Hb_0$ be a parity check $r_0\times n_0$ matrix of an $[n_0,n_0-r_0]_2R,\ell_0$ \emph{starting} code $C_0$,
\begin{equation}\label{eq32:H0}
  \Hb_0=[\hb_1\hb_2\ldots\hb_{n_0}],~\hb_j\in\F_2^{\,r_0},~j=1,\ldots,n_0,
\end{equation}
where $\hb_j$ is a binary $r_0$-positional column. Assume that $\Hb_0$ admits a starting $(R,\ell _0)$-partition $\Ps_0$ into $p(\Hb_0,\ell _0)$ subsets. Let $m\geq 1$ be an integer depending on $\Ps_0$ and $n_0$. With every column $\hb_j$ we associate an \emph{indicator} $\beta_j$ $\in \F_{2^m}\cup \{\ast\}$ so that $\beta_i\neq \beta_j$
if columns $\hb_i$ and $\hb_j$ belong to \emph{distinct }subsets of $ \Ps_0$. If $\hb_i$ and $\hb_j$ belong to the same subset we are free to assign either $\beta_i=\beta_j$ or $\beta_i\ne \beta_j$. Let $\Bs$ denote the \emph{indicator set} so that $\Bs=\{\beta_1,\beta _2,\ldots,\beta_{n_0}\}$; we  require $\#\Bs \geq p(\Hb_0,\ell _0)$.

Let $\D$ be an $(r_0+Rm)\times N_m$ matrix with \emph{the $r_0$ top rows equal to the zero vector} and with $N_m\le (R-\ell_0)\theta_{m,2}$. For $j=1,2,\dots,n_0$, we denote by  $\Ab(\hb_j,\beta_j)$ an $(r_0+Rm)\times 2^m$ matrix, in which the \emph{$r_0$ top rows are formed by $2^m$-fold repeating of the column $\hb_j$}. Let $\mathbf{0}_{v}$ be the zero matrix with $v$ rows where the number of columns (possibly, 1) is clear from the context.

Finally, define a new code $C$ as an $[n,n-(r_0+Rm)]_2R_C$ code with $n=2^m n_0+N_m$ and the parity check $(r_0+Rm)\times n$ matrix $\Hb_C$ of the following form:
\begin{align}
& \Hb_C \triangleq\left[ \D~\Ab(\hb_1,\beta_1)~\Ab(\hb_2,\beta_2)~\ldots ~
\Ab(\hb_{n_0},\beta_{n_0})\right],\label{eq32:QM-H}\db\\
&\Ab(\hb_j,\beta_j)\triangleq\left[ \renewcommand{\arraystretch}{1.1}
\begin{array}{cccc}
\hb_j & \hb_j & \mathbf{\cdots } & \hb_j \\
\xi _1 & \xi _2 & \cdots & \xi _{2^m} \\
\beta _j\xi _1 & \beta _j\xi _2 & \cdots & \beta _j\xi _{2^m} \\
\beta _j^{2}\xi _1 & \beta_j^{2}\xi _2 & \cdots & \beta _j^{2}\xi
_{2^m} \\
\vdots & \vdots & \vdots & \vdots \\
\beta _j^{R-1}\xi _1 & \beta_j^{R-1}\xi _2 & \cdots & \beta
_j^{R-1}\xi _{2^m}
\end{array}
\right] \T{if }\beta _j\in \F_{2^m},\label{eq32:Aj}\db\\
&\Ab(\hb_j,\beta_j)\triangleq\left[
\begin{array}{cccc}
\hb_j & \hb_j & \mathbf{\cdots } & \hb_j \\
\zb_{(R-1)m} & \zb_{(R-1)m} & \cdots & \zb_{(R-1)m} \\
\xi _1 & \xi _2 & \cdots & \xi _{2^m}
\end{array}
\right] \T{ if }\beta _j=\ast ,  \label{eq32:A*}\db\\
&\phantom{\Ab(\hb_j,\beta_j)\triangleq~}\{\xi _1,\xi _2,\ldots ,\xi _{2^m}\}=\F_{2^m},~\xi_1=0,~\xi_2=1. \label{eq32:xi}
\end{align}

If $m$, $\D$, and $\Bs$ are carefully chosen, then the covering radius $R_C$ of the new code $C$ is equal to the
covering radius $R$ of the starting code $C_0$; see  Sections \ref{sec4:R=2known}--\ref{sec7:R=3new}, where we use
the following notations:
\begin{align}
  & \Wb_m\T{ is a parity check $m\times(2^m-1)$ matrix of the binary Hamming code}\label{eq32:Wbm}\db\\
  &[2^m-1,2^m-1-m,3]_21;\notag\db\\
  &\Hcb_{2m}\T{ is a parity check }2m\times n_{2m} \T{ matrix of an }[n_{2m},n_{2m}-2m]_22\T{ code }
\Vc_{2m};\label{eq32:Hcb2m}\db\\
  &\QM_j^R\T{ is the $j$-th version of Construction QM for covering radius }R;\notag\db\\
  &\T{(the overline over QM is used to avoid confusion with notations from \cite{DMP-arXiv2025,DMP-IEEE2026})}.\notag
  \end{align}

\section{The known results on binary linear codes of covering radius 2 and 1-saturating sets}\label{sec4:R=2known}

We introduce the notation.
\begin{align}
&\phi(r)=\phi(2t)\triangleq27\cdot2^{r/2-4}-1 = 27\cdot2^{t-4}-1  \T{ for } r=2t;\label{eq4:notatKnwR=2even}\db\\
&\Phi(r)=\Phi(2t)\triangleq26\cdot2^{r/2-4}-1= 26\cdot2^{t-4}-1 \T{ for }r=10,18,20,~r=2t\ge28;\label{eq4:notatKnwR=2-r-even1}\db\\
&\widehat{\Phi}(r)=\widehat{\Phi}(2t)\triangleq26.5\cdot2^{r/2-4}-3= 26.5\cdot2^{t-4}-3 \T{ for } r=2t=22,24,26;\label{eq4:notatKnwR=2-r-even2}\db\\
&\vartheta(r)=\vartheta(2t -1)\triangleq5\cdot2^{(r-3)/2}-1= 5\cdot2^{t-2}-1   \T{ for }  r=2t -1.\label{eq4:notatKnwR=2odd}\db\\
&\D_1=\left[\begin{array}{c}
 \zb_{r_0+m}\\
 \Wb_m
 \end{array}\right].
\end{align}
In \cite{DMP-arXiv2025,DMP-IEEE2026}, the function $\vartheta(r)$ was denoted as $f(r)$. We use $\vartheta(r)$ here to avoid confusion with the constant $f(R)$ from \cite{Green}.

\begin{theorem}\label{th4:KnownFamR=2}
There exist infinite families of binary linear $[n,n-r]_2R$ covering codes with growing codimension $r$ and parameters as in \eqref{eq4:KnwFamR=2even r}--\eqref{eq4:KnwFamR=2odd r} and the sporadic covering codes with parameters as in \eqref{eq4:KnwSporadic}. The codes provide upper bounds on the length functions $\ell_2(r,2)$ and on the smallest sizes $s_2(r-1,1)$ of the $1$-saturating sets in $\PG(r-1,2)$ such that $\ell_2(r,2)=s_2(r-1,1)\le n$, see \eqref{eq1:cov=sat}. The code family \eqref{eq4:KnwFamR=2even r1} provides the upper bound on the constant for Green's Open Problem 40 \cite[Section 5, Problem 40]{Green} $f(2)\le1.32031$.
\begin{description}
  \item[(i)]\cite[p.\ 53]{DDL-ACCT3-1992},\cite[Example 3.1,\,equation (1.3)]{DDL-IEEE},\cite[Theorem 5.4.27(i),\,equation (5.4.30)]{CHLL-bookCovCod}
\begin{align}
&R=2,~r=2t\ge8,~t\ge4;~n=27\cdot2^{r/2-4}-1 =27\cdot2^{t-4}-1=\phi(r),
\label{eq4:KnwFamR=2even r}\db\\
&\overline{\mu}(2)\le1.4238;~\ell_2(r,2)=s_2(r-1,1)\le\phi(r).\notag
\end{align}

  \item[(ii)] \cite[Theorems 5.4, 5.7, equations (5.7), (5.8), (5.15)]{DMP-arXiv2025},\cite[Theorems 8, 11, equations (31), (32), (39)]{DMP-IEEE2026}
\begin{align}\label{eq4:KnwFamR=2even r1}
&R=2,~r=2t,~r=10,18,20,\T{ and }r\ge28,~t=5,9,10,\T{ and }t\ge14,\db\\
&n=26\cdot2^{r/2-4}-1= 26\cdot2^{t-4}-1=\Phi(r),~\overline{\mu}(2)\le1.32031; \notag\db\\
&\ell_2(r,2)=s_2(r-1,1)\le\Phi(r).\notag
\end{align}

  \item[(iii)] \cite[Theorem 1, equation (5)]{GabDavTombR=2}, \cite[Remark 3.1]{DDL-IEEE}, \cite[Theorem 5.4.27(ii)]{CHLL-bookCovCod}
\begin{align}
&R=2,~r=2t-1\ge3,~t\ge2,~n=5\cdot2^{(r-3)/2}-1=5\cdot2^{t-2}-1=\vartheta(r),\label{eq4:KnwFamR=2odd r}\db\\
&\overline{\mu}(2)\le1.5625;~\ell_2(r,2)=s_2(r-1,1)\le\vartheta(r).\notag
\end{align}

  \item[(iv)]\cite[Theorem 5.5]{DMP-arXiv2025},\cite[Theorem 9]{DMP-IEEE2026}
\begin{align}\label{eq4:KnwSporadic}
&R=2,~r=22,24,26,~n=26.5\cdot2^{r/2-4}-3=\widehat{\Phi}(r),~\mu(22,2)\le1.36956,\db\\
&\mu(24,2)\le1.37057,~\mu(26,2)\le1.37107;~\ell_2(r,2)=s_2(r-1,1)\le\widehat{\Phi}(r).\notag
\end{align}
\end{description}
\end{theorem}

\begin{theorem}\label{th4:-3}\cite[Theorem 5.1]{DMP-arXiv2025},\cite[Theorem 6]{DMP-IEEE2026}
$\mathbf{Construction~\overline{QM}_1^2}$. In the Basic Construction \emph{QM} of Section $\ref{sec3:qm concat}$, let the starting code $C_0$  be an $[n_0,n_0-r_0]_22$ code with a parity check matrix $\Hb_0$ of the form \eqref{eq32:H0} admitting a $2$-partition
into $p(\Hb_0)$ subsets. We define a new $[n,n-r]_2R_C$ code $C$ by a parity check
matrix $\Hb_C$ of the form \eqref{eq32:QM-H}--\eqref{eq32:Wbm}, where the indicator set $\Bs$, parameter $m$, and the
 submatrix $\D$  are as follows:
\begin{align}\label{eq4:-3inp}
&\QM_1^2:~\Bs\subseteq\{*\}\cup\F_{2^m}\setminus\{1\},~2^m\ge p(\Hb_0),~\D=
\left[\begin{array}{ccc}
 \zb_{r_0}&\zb_{r_0}&\zb_{r_0}\\
 \Wb_m\setminus\wb&\wb&\zb_m\\
 \zb_m&\wb&\Wb_m\setminus\wb
  \end{array}\right],
\end{align}
$\Wb_m\setminus\wb$ denotes the matrix obtained from $\Wb_m$ by deleting the column $\wb$.

Then the new code $C$ is a  binary linear $[n, n - r,3]_2R_C$ code with the parameters:
\begin{equation}\label{eq4:-3res}
\QM_1^2:~R_C=2,~n=2^m(n_0+2)-3,~r = r_0 + 2m.
\end{equation}
\end{theorem}

Note that the condition $2^m\ge p(\Hb_0)$ guarantees that distinct indicators from $\{*\}\cup\F_{2^m}\setminus\{1\}$ can be assigned to distinct subsets of the starting partition.

\begin{theorem}\label{th4:DDLR=2} \cite[Theorem 3.1]{DDL-IEEE}
$\mathbf{Construction~\QM_2^2}$. In the Basic Construction \emph{QM} of Section $\ref{sec3:qm concat}$, let the starting code $C_0$  be an $[n_0,n_0-r_0]_22$ code with a parity check matrix $\Hb_0$ of the form \eqref{eq32:H0} admitting a $2$-partition
into $p(\Hb_0)$ subsets.  We define a new $[n,n-r]_2R_C$ code $C$ by a parity check
$r\times n$ matrix $\Hb_C$ of the form \eqref{eq32:QM-H}, \eqref{eq32:Aj}, \eqref{eq32:xi}, \eqref{eq32:Wbm}, where the indicator set $\Bs$, parameter $m$, and the $(r_0+2m)\times ( 2^m-1)$ submatrix $\D$  are as follows:
\begin{equation}
\QM_2^2:~\Bs=\F_{2^m},~n_0\ge2^m\ge p(\Hb_0),~
 \D=\D_1. \label{eq4:CaseA2inp}
\end{equation}
Then the new code $C$ is an $[n, n - r]_2R_C$ code with the parameters:
\begin{equation}
\QM_2^2:~R_C=2,~n=2^m(n_0+1)-1,~r = r_0 + 2m,~p(\Hb_C)\le 2^{m+1}+1.\label{eq4:CaseA2res}
\end{equation}
\end{theorem}

\section{New results on binary linear codes of covering radius 2 and 1-saturating sets}\label{sec5:R=2new}
We introduce the notations.
\begin{align}
&\Theta(r)=\Theta(2t)\triangleq51\cdot2^{r/2-5}-1= 51\cdot2^{t-5}-1 \T{ for }r=10,18,20,~r=2t\ge28.\label{eq5:notEvenR=2New1}\db\\
&\widehat{\Theta}(r)\triangleq52\cdot2^{r/2-5}-3= 52\cdot2^{t-5}-3 \T{ for }r=2t=22,24,26.\label{eq5:notEvenR=2New2}
\end{align}

\begin{lemma}\label{lem5:Delta}
 Let $r=2t$. We have
 \begin{equation}\label{eq5:Delta}
  \Phi(r)-\Theta(r)=\widehat{\Phi}(r)-\widehat{\Theta}(r)=2^{r/2-5}=2^{t-5}.
 \end{equation}
\end{lemma}

\begin{proof}
  The assertions follow directly from \eqref{eq4:notatKnwR=2-r-even1}, \eqref{eq4:notatKnwR=2-r-even2}, \eqref{eq5:notEvenR=2New1}, \eqref{eq5:notEvenR=2New2}.
\end{proof}

In Theorems \ref{th5:Wu code} and \ref{th5:-3r=22,24,26}--\ref{th5:R=2NewInfFam}, let the covering density $\mu(r,2)$ and asymptotic covering density $\overline{\mu}(2)$ be as in $\eqref{eq1:CovDensity}$ and \eqref{eq1:liminfdens}. New $[n,n-r]_22$ codes, obtained in these theorems, provide upper bounds on the length functions $\ell_2(r,2)$ and on the smallest sizes $s_2(r-1,1)$ of the $1$-saturating sets in $\PG(r-1,2)$ such that $\ell_2(r,2)=s_2(r-1,1)\le n$, see \eqref{eq1:cov=sat}. The decrease of the known upper bounds on the length function $\ell_2(r,2)$ is $\Delta(r,2)$.

\begin{theorem}\label{th5:Wu code} (S. Wu \cite{Wu2026})
\begin{description}
  \item[(i)]
Let $C_{Wu}$ be the binary linear $[50,40]_2$ code with the parity check $10\times 50$ matrix $\Hb_{Wu}$ of the form
\begin{equation}\label{eq5:Hwu}
  \Hb_{Wu}\triangleq\left[
\begin{array}{*{10}{r}}
  1 & 2 & 4 & 15 & 16 & 32 & 65 & 86 & 128 & 173 \\
183 & 202 & 212 & 247 & 256 & 297 & 320 & 329 & 341 & 366 \\
373 & 381 & 391 & 403 & 438 & 460 & 479 & 491 & 502 & 559 \\
576 & 608 & 653 & 734 & 742 & 754 & 771 & 777 & 789 & 821 \\
846 & 855 & 869 & 881 & 893 & 897 & 927 & 981 & 1003 & 1004
\end{array}
\right],
\end{equation}
where each 10-positional binary column with the least significant binary digit on the 1-st position is represented by the corresponding decimal integer, for example, $1=[1000000000]^{tr},2=[0100000000]^{tr},15=[1111000000]^{tr},1004=[0011011111]^{tr}$.

Then $C_{Wu}$ is a $[50,40,3]_22$ code of covering radius $2$ and distance $3$; it provides  $\ell_2(10,2)=s_2(9,1)\le50=\Theta(10)$ and $\Delta(10,2)=1$, whereas the shortest previously known $[n,n-10]_22$ code is the $[51,41,3]_22$ code of \cite{KaikRoseADSlike2003}. The covering density is $\mu(C_{Wu})=319/256=1.24609375$. The code $C_{Wu}$ is locally optimal (LO).  The matrix $\Hb_{Wu}$ corresponds to a minimal $1$-saturating $50$-set in $\PG(9,2)$.

  \item[(ii)]
  Let the partition $\Ps_{Wu}$ of the column set of $\Hb_{Wu}$ into $10$ subsets be as follows:
\begin{align}\label{eq5:PWu}
&\Ps_{Wu}\triangleq\{2,128,202,212,771,855,897,981\}, \{86\}, \{381,893,1003\},\db\\
&\{183,297\},\{1,65,247,256,320,438,502,734\},\{15,173,329,366,460,559,653,846\},\db\notag\\
&\{4,16,391,491,742,754,869,881\},\{479,1004\},\{403\},\db\notag\\
&\{32,341,373,576,608,777,789,821,927\}.\notag
\end{align}
Then $\Ps_{Wu}$ is a $2$-partition of the matrix $\Hb_{Wu}$, i.e. $p(\Hb_{Wu};\Ps_{Wu})=10$.

  \item[(iii)]
There is the linearly dependent triple of columns $(734,491,821)$ from three distinct subsets ($5$-th,$7$-th,$10$-th) of the partition $\Ps_{Wu}$.
\end{description}
\end{theorem}

\begin{proof}
The assertions were obtained by S. Wu using computer search, see \cite[Section~3]{Wu2026}. The results were independently checked by him using two verification programs with opposite algorithms, see\cite[Section 4]{Wu2026}.
\end{proof}

\begin{remark}
In order to have confirmation of the results of Theorem \ref{th5:Wu code} by two different people using different programs, S. Marcugini verified the statements of Theorem \ref{th5:Wu code} using programs in the MAGMA language [2].
\end{remark}

\begin{theorem}\label{th5:-3r=22,24,26}
There exist three new  $[n,n-r,3]_22$ codes with the following parameters and bounds:
\begin{align}\label{eq5:NewSporadic}
&R=2,~r=22,24,26,~n=52\cdot2^{r/2-5}-3=\widehat{\Theta}(r);\\
&\mu(22,2)\le1.31833,~\mu(24,2)\le1.31933,~\mu(26,2)\le1.31982;\notag\db\\
&\ell_2(r,2)=s_2(r-1,1)\le52\cdot2^{r/2-5}-3=\widehat{\Theta}(r),~\Delta(r,2)=2^{r/2-5}.\notag
\end{align}
\end{theorem}

\begin{proof}
  We apply Construction $\QM_1^2$ from Theorem \ref{th4:-3}.  Let the starting code $C_0$  be the $[50,40,3]_22$ code $C_{Wu}$ with the parity check matrix $\Hb_{Wu}=\Hb_0$, see Theorem \ref{th5:Wu code}. So, $n_0=50$, $r_0=10$, $\Ps_0=\Ps_{Wu}$, $p(\Hb_0;\Ps_0)=10$. We take $m=6,7,8$, respectively, and obtain the required codes by \eqref{eq4:-3res}. The decrease  $\Delta(r,2)$ follows from \eqref{eq4:KnwSporadic}, \eqref{eq5:Delta}.
\end{proof}

\begin{theorem}\label{th5:r=18}
$\mathbf{Construction~\QM_3^2}$. In the Basic Construction \emph{QM} of Section $\ref{sec3:qm concat}$, let the $[n_0,n_0-r_0]_2R$ starting code $C_0$  be the $[50,40]_22$ code $C_{Wu}$ with the parity check matrix $\Hb_{Wu}=\Hb_0$, see \eqref{eq32:H0} and  Theorem $\ref{th5:Wu code}$ with \eqref{eq5:Hwu}. We define a new $[n,n-r,d_C]_2R_C$ code $C$ by a parity check
 matrix $\Hb_C$ of the form \eqref{eq32:QM-H}, \eqref{eq32:Aj}, \eqref{eq32:xi}, \eqref{eq32:Wbm}, where the indicator set $\Bs$, parameter $m$, and the submatrix $\D$  are as follows: $\Bs=\F_{2^m},~m=4$, and $\D=\D_1$.
Then the new code $C$ is an $[n,n-r,3]_2R_C$ code with the following parameters and bounds:
\begin{align}
&\QM_3^2:~R_C=2,~n=\Theta(18)=815,~r=18,~\mu(18,2)\le1.26847,\label{eq5:QM32res}\db\\
&p(\Hb_C)\le 2^5+1;~\ell_2(18,2)=s_2(17,1)\le815=\Theta(18),~\Delta(18,2)=2^4.\notag
\end{align}
Also, the matrix $\Hb_C$ admits a $2$-partition $\Ps_C$ into $2^5+1$ subsets such that
there are three subsets, consisting of one column, and these columns form a linearly dependent triple, see~\eqref{eq5:3subsets}.
\end{theorem}

\begin{proof}
 We obtain a 2-partition $\Ps_{Wu}^*$ into $10+3\cdot2=16=2^m$ subsets, by splitting each of the fifth, seventh, and tenth
subsets of $\Ps_{Wu}$ \eqref{eq5:PWu} into three subsets with the three linearly dependent columns 734, 491, 821 placed in three distinct singleton subsets. We then set $\Ps_0=\Ps_{Wu}^*$ and assign to all columns of every subset of $\Ps_0$ the same indicator different from the indicators of the other subsets; this provides $\Bs=\F_{2^m}$.

In \eqref{eq5:QM32res}, the values of $n$ and $r$ follow from the construction. The minimum distance $d_C=3$ follows from the submatrix $\D$.

By Definition~\ref{Def1:CoverRad}(ii), to prove $R_C=2$ we show that every vector $\Ub \in\F_2^{\,r}$  can be represented as a sum of $\le2$ columns of $\Hb_C$. Let $\Ub=(\pib,\ub_1,\ub_2)^{tr}\in\F_2^{\,r}$, where $\pib\in\F_2^{\,r_0},\ub_1,\ub_2\in\F_2^m$. By \eqref{eq32:H0}, $\Hb_0=[\hb_1\hb_2\ldots\hb_{50}],~\hb_j\in\F_2^{\,r_0},~j=1,\ldots,50$.

  Let $\pib=\hb_i+\hb_j$. Then $\hb_i$, $\hb_j$ belong to distinct subsets of $\Ps_0$; $\beta_i\ne\beta_j$.
   We solve the system $\xb+ \yb=\ub_1$, $\beta_i \xb+\beta_j \yb=\ub_2$, for $\xb,\yb\in\F_{2^m}$. Since $\beta_i\ne\beta_j$, the coefficient matrix is nonsingular over $\F_{2^m}$, and the system has a unique solution.
  Now $\Ub= (\hb_i,\xb,\beta_i \xb)^{tr} +(\hb_j,\yb,\beta_j \yb)^{tr}$.

  Let $\pib=\hb_j$. Then
  $\Ub=(\hb_j, \ub_1, \beta_j\ub_1)^{tr} +(\zb_{r_0+m}, \beta_j\ub_1+\ub_2)^{tr}$, $(\zb_{r_0+m}, \beta_j\ub_1+\ub_2)^{tr}\in\D$.

  Let $\pib=\zb_{r_0}$.
  Let $\ub_1\ne0$. As $\Bs=\F_{2^m}$, indicator $\beta_k = \ub_2/\ub_1$
 exists; now $\Ub=(\hb_k, \zb_m, \zb_m)^{tr} +(\hb_k, \ub_1, \beta_k\ub_1)^{tr}$. If $\ub_1=0$, then $\Ub=(\zb_{r_0+m}, \ub_2)^{tr}\in\D$.

 Finally, we construct the $2$-partition $\Ps_{C}$ of $2^{m+1}+1=2^5+1$ subsets.

We obtain $2^{m}=16$ interim subsets, partitioning the columns of $\Hb_C$ (except $\D$) so that if columns $\hb_i, \hb_j$ of $\Hb_0=\Hb_{Wu}$ belong to distinct subsets of
$\Ps_0=\Ps_{Wu}^*$, then the columns of submatrices $\Ab(\hb_i,\beta_i)$, $\Ab(\hb_j,\beta_j)$ belong to distinct subsets
of the partition of $\Hb_C$.
Then we obtain $2^{m+1}=2^5$ subsets of $\Ps_{C}$ partitioning every interim subset into two
subsets such that the first one contains all columns of the
form $(\hb_\bullet, \zb_4, \zb_4)^{tr}$, where $\xi_1=0$ is used, and the second one contains the rest of the
columns. This separation corresponds to the case  $\pib=\zb_{r_0}$, $\ub_1\ne0$. The  $(2^5+1)$-st subset consists of the columns of $\D$.

By Theorem \ref{th5:Wu code}(ii),(iii) and the processes forming $\Ps_{Wu}^*$ and $\Ps_{C}$, there are three subsets of $\Ps_{C}$, consisting of one column, such that
\begin{equation}\label{eq5:3subsets}
\{(491,\zb_8)^{tr}\},\{(734,\zb_8)^{tr}\},\{(821,\zb_8)^{tr}\};~(491,\zb_8)^{tr}=(734,\zb_8)^{tr}+(821,\zb_8)^{tr}.
\end{equation}
 \end{proof}

\begin{theorem}\label{th5:r=28}
$\mathbf{Construction~\QM_4^2}$. In the Basic Construction \emph{QM} of Section $\ref{sec3:qm concat}$, let the $[n_0,n_0-r_0]_2R$ starting code $C_0$  be the $[815,815-18,3]_22$ code obtained by Construction $\QM_3^2$. We define a new $[n,n-r,d_C]_2R_C$ code $C$ by a parity check matrix $\Hb_C$ of the form \eqref{eq32:QM-H}--\eqref{eq32:Wbm}, where the indicator set $\Bs$, parameter $m$, and the submatrix $\D$  are as follows: $\Bs=\F_{2^m}\cup\{*\},~m=5,~
\D=\D_1$.
Then the new code $C$ is an $[n,n-r,3]_2R_C$ code with the following parameters and bounds:
\begin{align}
&\QM_4^2:~R_C=2,~n=\Theta(28)=26111,~r=28,~\mu(28,2)\le1.26998,\label{eq5:QM42res}\db\\
&p(\Hb_C)\le 2^6+2;~\ell_2(28,2)=s_2(27,1)\le26111=\Theta(28),~\Delta(28,2)=2^9.\notag
\end{align}
\end{theorem}

\begin{proof}
   The parity check matrix $\Hb_C$ and the 2-partition $\Ps_C$ from Theorem \ref{th5:r=18} and its proof are the parity check $18\times 815$ matrix $\Hb_0$ and the 2-partition $\Ps_0$ of $2^5+1$ subsets for this theorem. By \eqref{eq32:H0}, $\Hb_0=[\hb_1\hb_2\ldots\hb_{815}],~\hb_j\in\F_2^{18},~j=1,\ldots,815$.
Let the column $(491,\zb_8)^{tr}$ \eqref{eq5:3subsets} be the column $\hb_a\in\Hb_0$, $a\in\{1,2,\ldots,815\}$.
 We assign to all columns of every subset of $\Ps_0$ the same indicator other than the indicators of other subsets; at the same time, we set $\beta_a=*$. This provides $\Bs=\F_{2^m}\cup\{*\}$.

 In \eqref{eq5:QM42res}, the values of $n$ and $r$ follow from the construction. The minimum distance $d_C=3$ follows from the submatrix $\D$.

 By Definition~\ref{Def1:CoverRad}(ii), to prove $R_C=2$ we show that every vector $\Ub \in\F_2^{\,r}$  can be represented as a sum of $\le2$ columns of $\Hb_C$. Let $\Ub=(\pib,\ub_1,\ub_2)^{tr}\in\F_2^{\,r}$, where $\pib\in\F_2^{\,r_0},\ub_1,\ub_2\in\F_2^m$.

Let $\pib=\hb_i+\hb_j$.
Then $\hb_i$ and $\hb_j$ belong to distinct subsets of $\Ps_0$; $\beta_i\ne\beta_j$.
If $\beta_i,\beta_j\ne *$, we find $\xb$, $\yb$ from the system $\xb+ \yb=\ub_1$, $\beta_i \xb+\beta_j \yb=\ub_2$.
  Now $\Ub= (\hb_i,\xb,\beta_i \xb)^{tr} +(\hb_j,\yb,\beta_j \yb)^{tr}$.
If $j=a$ and $\beta_a=*$, then $\Ub= (\hb_i,\ub_1,\beta_i\ub_1)^{tr} +(\hb_a, \zb_m,\beta_i\ub_1+\ub_2)^{tr}$.

Let $\pib=\hb_j$.
  If $j\ne a$, $\beta_j\ne *$, then $\Ub=(\hb_j, \ub_1, \beta_j\ub_1)^{tr} +(\zb_{r_0+m}, \beta_j\ub_1+\ub_2)^{tr}$ with $(\zb_{r_0+m}, \beta_j\ub_1+\ub_2)^{tr}\in\D$. If $j=a$, then, by \eqref{eq5:3subsets}, $\pib=(734,\zb_8)^{tr}+(821,\zb_8)^{tr}$. These two columns belong to distinct subsets of $\Ps_0$ and we are again in the case $\pib=\hb_i+\hb_j$ with $\beta_i\ne\beta_j$.

   $\pib=\zb_{r_0}$.
  Let $\ub_1\ne0$. As $\Bs=\F_{2^m}\cup\{*\}$, indicator $\beta_k = \ub_2/\ub_1\ne*$
 exists. We have $\Ub=(\hb_k, \zb_m, \zb_m)^{tr} +(\hb_k, \ub_1, \beta_k\ub_1)^{tr}$.
If $\ub_1=0$, then $\Ub=(\zb_{r_0+m},\ub_2)^{tr}\in\D$.

Now, we construct the $2$-partition $\Ps_C$ of $2^{m+1}+2=2^6+2$ subsets.
We obtain $2^{m}+1=33$ interim subsets, partitioning the columns of $\Hb_C$ (except $\D$) so that if columns $\hb_i, \hb_j$ of $\Hb_0$ belong to distinct subsets of $\Ps_0$, then the columns of submatrices $\Ab(\hb_i,\beta_i)$, $\Ab(\hb_j,\beta_j)$ belong to distinct subsets of the partition of $\Hb_C$. Then we obtain $2^{m+1}=2^6$ subsets of $\Ps_{C}$ splitting every interim subset (except for $\{\Ab(\hb_a,*)\}$) into two subsets such that the first one contains all columns of the
form $(\hb_\bullet, \zb_5, \zb_5)^{tr}$, where $\xi_1=0$ is used, and the second one contains the rest of the
columns. This is connected with the case  $\pib=\zb_{r_0}$, $\ub_1\ne0$, in which the columns of $\{\Ab(\hb_1,*)\}$ do not appear; these columns form $(2^6+1)$-st subset. The  $(2^6+2)$-th subset consists of the columns of~$\D$.
\end{proof}

\begin{theorem}\label{th5:n0Theta}
 $\mathbf{Construction~\QM_5^2}$. In Construction $\QM_2^2$ of Theorem \emph{\ref{th4:DDLR=2}}, let  the starting code $C_0$ be an $[n_0,n_0-2t_0,3]_22$ code with $r_0=2t_0$, $n_0=51\cdot2^{r_0/2-5}-1=51\cdot2^{t_0-5}-1=\Theta(r_0)$, and a parity check matrix $\Hb_0$ with $p(\Hb_0)\le2^{\lambda_0}+2$. We define a new $[n,n-r,d_C]_2R_C$ code $C$ by a parity check
$r\times n$ matrix $\Hb_C$ of the form \eqref{eq32:QM-H}, \eqref{eq32:Aj}, \eqref{eq32:xi}, \eqref{eq32:Wbm} with the following indicator set $\Bs$, parameter $m$, and submatrix~$\D$:
\begin{equation}
\QM_5^2:~\Bs=\F_{2^m},~t_0\ge m\ge \lambda_0+1,~
 \D=\D_1. \label{eq5:QM1ainp}
\end{equation}
Then the new code $C$ is an $[n, n - r,3]_2R_C$ code with the following parameters and bounds:
\begin{align}
&\QM_5^2:~R_C=2,~r = 2t_0+2m=2t,~t=t_0+m,~n=51\cdot2^{r/2-5}-1=\Theta(r),\label{eq5:QM1ares}\db\\
&p(\Hb_C)\le 2^{m+1}+1;~\ell_2(r,2)=s_2(r-1,1)\le\Theta(r),~\Delta(r,2)=2^{r/2-5}.  \notag
\end{align}
\end{theorem}

\begin{proof}
For $n_0=\Theta(2t_0)$ and $p(\Hb_0)\le2^{\lambda_0}+2$, the condition $n_0\ge2^m\ge p(\Hb_0)$  of \eqref{eq4:CaseA2inp} holds if $t_0\ge m\ge \lambda_0+1$. So, we can use the results of
\eqref{eq4:CaseA2res}. In particular, we have $n=2^m(\Theta(2t_0)+1)-1=51\cdot2^{t_0-5+m}-1=\Theta(2t_0+2m)=\Theta(2t)=\Theta(r)$. By Lemma \ref{lem5:Delta}, $\Delta(r,2)=\Phi(r)-\Theta(r)=2^{r/2-5}$, see \eqref{eq4:KnwFamR=2even r1}, \eqref{eq5:Delta}. The minimum distance $d_C=3$ follows from $\D$.
\end{proof}

In Tables \ref{tab5:LengFun r32R=2} and \ref{tab5:LengFun33r64R=2}, the parameters $n,r$, and the covering density \eqref{eq1:CovDensity} of the best (as far as the authors know) binary linear covering $[n,n-r]_22$ codes of covering radius $R=2$ and redundancy (codimension) $2\le r\le64$  are written. The listed values of $n$ provide upper bounds on the length function $\ell_2(r,2)$ and on the smallest sizes $s_2(r-1,1)$ of the $1$-saturating sets in $\PG(r-1,2)$ so that $\ell_2(r,2)=s_2(r-1,1)\le n$; the asterisk * indicates an exact value of $\ell_2(r,2)$. The new results, obtained in this paper, are written in bold font and also are marked by the big star $\bigstar$. For the known $n$ we write the references (the columns ``references'' and  ``refer.''). For the new codes we give the size $p(\Hb)$ of a 2-partition of a parity check matrix $\Hb$  and the decrease $\Delta(r,2)$ of the known upper bounds on $\ell_2(r,2)$, provided by the new results.
The new results are obtained using Theorems \ref{th5:Wu code}, \ref{th5:-3r=22,24,26}--\ref{th5:n0Theta} (Table \ref{tab5:LengFun r32R=2}) and only Theorem \ref{th5:n0Theta} (Table \ref{tab5:LengFun33r64R=2}). For the starting code we give the values of $n_0$ and $m$, and (only in Table \ref{tab5:LengFun r32R=2}) in the column ``Th.'' we list the theorems used.

The tables illustrate the iterative process to obtain the infinite family of Theorem \ref{th5:R=2NewInfFam}.

\begin{table}[htbp]
\caption{The parameters of the best binary linear $[n,n-r]_22$ covering codes, $r\le32$,
such that $\ell_2(r,2)=s_2(r-1,1)\le n$; the values of $n$ follow from $\phi(r),\vartheta(r)$ \eqref{eq4:notatKnwR=2even}, \eqref{eq4:notatKnwR=2odd}, \eqref{eq4:KnwFamR=2even r}, \eqref{eq4:KnwFamR=2odd r} (the known values) and $\Theta(r),\widehat{\Theta}(r)$ \eqref{eq5:notEvenR=2New1}, \eqref{eq5:notEvenR=2New2}, \eqref{eq5:NewSporadic}, \eqref{eq5:QM32res}, \eqref{eq5:QM42res}, \eqref{eq5:QM1ares} (the new ones, bold font); $\Delta=\Delta(r,2)$}
 \centering
  \begin{tabular}
  {r|c|r|c|l|c|c||c|c|c}\hline
$r$&for $n$ &\multicolumn{1}{c|}{$n$}&references&density$\le$&$p(\Hb)\le$&$\Delta$&$n_0$&$m$&Th.\\\hline
2 &&$2^{*}$&\cite{CHLL-bookCovCod,GrahSlo}&1      &    &&&\\
3 &$\vartheta(3)$&$4^{*}$&\cite{CHLL-bookCovCod,GrahSlo}&1.37500& &&&\\
4 &&$5^{*}$&\cite{CHLL-bookCovCod,GrahSlo}&1      & &&&\\
5 &$\vartheta(5)$&$9^{*}$&\cite{CHLL-bookCovCod,GrahSlo}&1.43750&    &&&\\
6 &&$13^{*}$&\cite{CHLL-bookCovCod,GrahSlo,Struik}&1.43750&    &&&\\
7 &$\vartheta(7)$&$19^{*}$ &\cite{DDL-IEEE,GabDavTombR=2,GrahSlo}&1.49219&&&&\\\hline
8 &$\phi(8)$&26  &\cite{BrPl1990,DDL-IEEE,GabDavTombR=2}&1.37500&&&&\\
9 &$\vartheta(9)$&39     &\cite{GabDavTombR=2}    &1.52540  &&&&\\
10&$\mathbf{\Theta(10)}$&    \textbf{50}        &$\bigstar$  &\textbf{1.24610}&\textbf{10}&\textbf{1}&&&\ref{th5:Wu code}\\
11&$\vartheta(11)$& 79    &\cite{DDL-IEEE,GabDavTombR=2}&1.54346& &       &&\\
12&$\phi(12)$&107& \cite{DDL-IEEE}&1.41089&&&&\\
13&$\vartheta(13)$&159   & \cite{GabDavTombR=2}    &1.55286&&        &&\\
14&$\phi(14)$&215& \cite{DDL-IEEE}&1.41730&&&&\\
15&$\vartheta(15)$&319    &\cite{DDL-IEEE,GabDavTombR=2}&1.55765&&          &&\\
16&$\phi(16)$&431 &\cite{DDL-IEEE}&1.42055&&&&\\
17&$\vartheta(17)$&639    &\cite{GabDavTombR=2}&1.56007&&            &&\\\hline
18&$\mathbf{\Theta(18)}$&\textbf{815}&$\bigstar$&\textbf{1.26847}&$\mathbf{2^5+1}$&$\mathbf{2^4}$&$\Theta(10)$&4&\ref{th5:r=18}\\
19&$\vartheta(19)$&1279 &\cite{DDL-IEEE,GabDavTombR=2}&1.56129&&&&&\\
20&$\mathbf{\Theta(20)}$&\textbf{1631}&$\bigstar$&\textbf{1.26925}&$\mathbf{2^6+1}$&$\mathbf{2^5}$&$\Theta(10)$&5&\ref{th5:n0Theta}\\
21&$\vartheta(21)$& 2559 &\cite{DDL-IEEE,GabDavTombR=2}&1.56190&  &&&\\
22&$\mathbf{\widehat{\Theta}(22)}$&\textbf{3325}&$\bigstar$&\textbf{1.31833}&&$\mathbf{2^6}$&$\Theta(10)$&6&\ref{th5:-3r=22,24,26}\\
23&$\vartheta(23)$& 5119 &\cite{DDL-IEEE,GabDavTombR=2}&1.56220&  &&&\\
24&$\mathbf{\widehat{\Theta}(24)}$&\textbf{6653}&$\bigstar$&\textbf{1.31933}&&$\mathbf{2^7}$&$\Theta(10)$&7&\ref{th5:-3r=22,24,26}\\
25&$\vartheta(25)$&10239 &\cite{GabDavTombR=2}&1.56235&& &&     \\
26&$\mathbf{\widehat{\Theta}(26)}$&\textbf{13309} &$\bigstar$&\textbf{1.31982}&&$\mathbf{2^8}$&$\Theta(10)$&8&\ref{th5:-3r=22,24,26}\\
27&$\vartheta(27)$&20479&\cite{DDL-IEEE,GabDavTombR=2}&1.56243&&&&\\
28&$\mathbf{\Theta(28)}$&\textbf{26111}&$\bigstar$ &\textbf{1.26998}&$\mathbf{2^6+2}$&$\mathbf{2^9}$&$\Theta(18)$&5&\ref{th5:r=28}\\
29&$\vartheta(29)$&40959&\cite{DDL-IEEE,GabDavTombR=2}&1.56247&&&&\\
30&$\mathbf{\Theta(30)}$&\textbf{52223}&$\bigstar$ &\textbf{1.27000}&$\mathbf{2^7+1}$&$\mathbf{2^{10}}$&$\Theta(18)$&6&\ref{th5:n0Theta}\\
31&$\vartheta(31)$&81919&\cite{DDL-IEEE,GabDavTombR=2}&1.56249&&&&\\
32&$\mathbf{\Theta(32)}$&\textbf{104447}&$\bigstar$&\textbf{1.27001}&$\mathbf{2^8+1}$&$\mathbf{2^{11}}$&$\Theta(18)$&7&\ref{th5:n0Theta}
\\\hline
 \end{tabular}
\label{tab5:LengFun r32R=2}
\end{table}

\begin{table}[htbp]
\caption{The parameters of the best binary linear $[n,n-r]_22$ covering codes,
$33\le r\le64$, such that $\ell_2(r,2)\le n$; the values of $n$ follow from $\vartheta(r)$ \eqref{eq4:notatKnwR=2odd}, \eqref{eq4:KnwFamR=2odd r} (the known values) and
$\Theta(r)$ \eqref{eq5:notEvenR=2New1}, \eqref{eq5:QM1ares} (the new ones, bold font); $\Delta=\Delta(r,2)$}
 \centering
  \begin{tabular}
{r|c|r|c|r|c|c||c|r}\hline
$r$&for $n$&\multicolumn{1}{c|}{$n$}&refer.&density$\le$&$p(\Hb)\le$&$\Delta$&$n_0$&$m$\\\hline
33&$\vartheta(33)$&163839   &\cite{GabDavTombR=2}&1.56250&&&\\
34&$\mathbf{\Theta(34)}$&\textbf{208895}&$\bigstar$&\textbf{1.27002}&$\mathbf{2^8+1}$&$\mathbf{2^{12}}$&$\Theta(20)$&7\\
35&$\vartheta(35)$&327679   &\cite{DDL-IEEE,GabDavTombR=2}&1.56250&&&\\
36&$\mathbf{\Theta(36)}$&\textbf{417791}&$\bigstar$&\textbf{1.27002}&$\mathbf{2^9+1}$&$\mathbf{2^{13}}$&$\Theta(20)$&8\\
37&$\vartheta(37)$&655359   &\cite{DDL-IEEE,GabDavTombR=2}&1.56250&&&\\
38&$\mathbf{\Theta(38)}$&\textbf{835583}&$\bigstar$&\textbf{1.27002}&$\mathbf{2^{10}+1}$&$\mathbf{2^{14}}$&$\Theta(20)$&9\\
39&$\vartheta(39)$&1310719  &\cite{DDL-IEEE,GabDavTombR=2}&1.56250&&&\\
40&$\mathbf{\Theta(40)}$&\textbf{1671167}&$\bigstar$&\textbf{1.27002}&$\mathbf{2^{11}+1}$&$\mathbf{2^{15}}$&$\Theta(20)$&10\\
41&$\vartheta(41)$&  2621439&\cite{DDL-IEEE,GabDavTombR=2}&1.56250&&&\\
42&$\mathbf{\Theta(42)}$&\textbf{3342335}&$\bigstar$&\textbf{1.27002}&$\mathbf{2^8+1}$&$\mathbf{2^{16}}$&$\Theta(28)$&7\\
43&$\vartheta(43)$& 5242879&\cite{DDL-IEEE,GabDavTombR=2}&1.56250&&&\\
44&$\mathbf{\Theta(44)}$&\textbf{6684671}&$\bigstar$&\textbf{1.27002}&$\mathbf{2^9+1}$&$\mathbf{2^{17}}$&$\Theta(28)$&8\\
45&$\vartheta(45)$& 10485759&\cite{DDL-IEEE,GabDavTombR=2}&1.56250&&&\\
46&$\mathbf{\Theta(46)}$&\textbf{13369343}&$\bigstar$&\textbf{1.27002}&$\mathbf{2^{9}+1}$&$\mathbf{2^{18}}$&$\Theta(30)$&8\\
47&$\vartheta(47)$& 20971519&\cite{DDL-IEEE,GabDavTombR=2}&1.56250&&&\\
48&$\mathbf{\Theta(48)}$&\textbf{26738687}&$\bigstar$ &\textbf{1.27002}&$\mathbf{2^{10}+1}$&$\mathbf{2^{19}}$&$\Theta(30)$&9\\
49&$ \vartheta(49)$&41943039&\cite{DDL-IEEE,GabDavTombR=2}&1.56250&&&\\
50&$\mathbf{\Theta(50)}$&\textbf{53477375}&$\bigstar$ &\textbf{1.27002}&$\mathbf{2^{10}+1}$&$\mathbf{2^{20}}$&$\Theta(32)$&9\\
51&$\vartheta(51)$& 83886079&\cite{DDL-IEEE,GabDavTombR=2}&1.56250&&&\\
52&$\mathbf{\Theta(52)}$&\textbf{106954751}&$\bigstar$&\textbf{1.27002}&$\mathbf{2^{11}+1}$&$\mathbf{2^{21}}$
&$\Theta(32)$&10\\
53&$\vartheta(53)$&167772159&\cite{DDL-IEEE,GabDavTombR=2}&1.56250&&&\\
54&$\mathbf{\Theta(54)}$&\textbf{213909503}&$\bigstar$&\textbf{1.27002}&$\mathbf{2^{11}+1}$&$\mathbf{2^{22}}$&$\Theta(34)$&10\\
55&$\vartheta(55)$&335544319&\cite{DDL-IEEE,GabDavTombR=2}&1.56250&&&\\
56&$\mathbf{\Theta(56)}$&\textbf{427819007}&$\bigstar$&\textbf{1.27002}&$\mathbf{2^{12}+1}$&$\mathbf{2^{23}}$&$\Theta(34)$&11\\

57&$\vartheta(57)$&671088639&\cite{DDL-IEEE,GabDavTombR=2}&1.56250&&&\\
58&$\mathbf{\Theta(58)}$&\textbf{855638015}&$\bigstar$&\textbf{1.27002}&$\mathbf{2^{12}+1}$&$\mathbf{2^{24}}$&$\Theta(36)$&11\\
59&$\vartheta(59)$&1342177279&\cite{DDL-IEEE,GabDavTombR=2}&1.56250&&&\\
60&$\mathbf{\Theta(60)}$&\textbf{1711276031}&$\bigstar$&\textbf{1.27002}&$\mathbf{2^{13}+1}$&$\mathbf{2^{25}}$&$\Theta(36)$&12\\

61&$\vartheta(61)$&2684354559&\cite{DDL-IEEE,GabDavTombR=2}&1.56250&&&\\
62&$\mathbf{\Theta(62)}$&\textbf{3422552063}&$\bigstar$&\textbf{1.27002}&$\mathbf{2^{13}+1}$&$\mathbf{2^{26}}$&$\Theta(38)$&12\\

63&$\vartheta(63)$&5368709119&\cite{DDL-IEEE,GabDavTombR=2}&1.56250&&&\\
64&$\mathbf{\Theta(64)}$&\textbf{6845104127}&$\bigstar$&\textbf{1.27002}&$\mathbf{2^{14}+1}$&$\mathbf{2^{27}}$&$\Theta(38)$&13\\\hline
 \end{tabular}
  \label{tab5:LengFun33r64R=2}
\end{table}

\begin{theorem}\label{th5:R=2NewInfFam}
There exists a new infinite family of binary linear  $[n,n-r,3]_22$ codes
with a parity check matrix $\Hb$ of the form \eqref{eq32:QM-H}--\eqref{eq32:Wbm}, and parameters and bounds as in \eqref{eq5:InfFamR=2New}.
\begin{align}\label{eq5:InfFamR=2New}
&R=2,~r=2t\ge28\T{ and }r=10,18,20,~t=5,9,10,\T{ and }t\ge14, \db\\
&n=51\cdot2^{r/2-5}-1= 51\cdot2^{t-5}-1=\Theta(r),~\overline{\mu}(2)\le\frac{2601}{2048}<1.27002;\notag\db\\
&\ell_2(r,2)=s_2(r-1,1)\le51\cdot2^{r/2-5}-1=\Theta(r),~\Delta(r,2)=2^{r/2-5}. \notag
\end{align}
Also, this infinite family provides the upper bound $f(2)\le2601/2048<1.27002$ on the constant $f(2)$ of the Green's Open Problem 40 \cite[Section 5, Problem 40]{Green}.

Let $p(\Hb)$ be the size of a $2$-partition of the parity check matrix $\Hb$. For even $r=10$, $18$, $20$, $28\le r\le40$, the sizes $p(\Hb)$ are given in Tables \emph{\ref{tab5:LengFun r32R=2}, \ref{tab5:LengFun33r64R=2}}. For even $r\ge42$, we have
\begin{equation}\label{eq5:InfFamR=2NewLenFunLambda}
 \lambda=\lfloor r/4\rfloor-2,~p(\Hb)=2^\lambda+1,~\lambda+1=t-(\lceil r/4\rceil+1)\le t-12.
\end{equation}
\end{theorem}

\begin{proof}
We prove the assertions via an iterative process; it is illustrated by Tables \ref{tab5:LengFun r32R=2}, \ref{tab5:LengFun33r64R=2}.

\textbf{(i)} For $r=10$, $t=5$, we have \cite{Wu2026} the new $[50,40,3]_22=[\Theta(10),\Theta(10)-10,3]_22$ code $C_{Wu}$ with the parity check matrix $\Hb_{Wu}$ \eqref{eq5:Hwu}, admitting the new 2-partition $\Ps_{Wu}$ \eqref{eq5:PWu} into 10 subsets, see  Theorem \ref{th5:Wu code}. The code and the partition serve as the seed for the subsequent iterative process.

\textbf{(ii)} For $r=18,28$, Theorems \ref{th5:r=18}, \ref{th5:r=28} with \eqref{eq5:QM32res}--\eqref{eq5:QM42res} obtain $[\Theta(18),\Theta(18)-18,3]_22$ and $[\Theta(28),\Theta(28)-28,3]_22$ codes with $p(\Hb)=2^5+1$ and $p(\Hb)=2^6+2$, respectively.

\textbf{(iii)} For even $r=20$ and $r\ge30$, we iteratively use Construction $\QM_5^2$ of Theorem~\ref{th5:n0Theta}
 with \eqref{eq5:QM1ainp}, \eqref{eq5:QM1ares},
choosing convenient starting codes $C_0$ and parameters $m$.

\textbf{(iii-1)} As the starting code $C_0$ we take the $[50,40,3]_22$ code $C_{Wu}$ with $n_0=\Theta(10)$, $r_0=10$, $t_0=5$, $p(\Hb_0)=10<2^4+1$, see Theorem \ref{th5:Wu code}; we have $\lambda_0=4$. The condition $t_0\ge m\ge \lambda_0+1$ from \eqref{eq5:QM1ainp} holds if $m=4,5$. We take $m=5$. By \eqref{eq5:QM1ares}, we obtain a $[\Theta(20),\Theta(20)-20]_22$ code with $r=20$, $t=10$, $p(\Hb)=2^6+1$, $\Delta(20,2)=2^5$.

\textbf{(iii-2)} Let $C_0$ be the $[\Theta(18),\Theta(18)-18,3]_22$ code with $r_0=18$, $t_0=9$, $p(\Hb)=2^5+1$. We have $t_0\ge m\ge \lambda_0+1$ if $m=6,7,8,9$. We take $m=6,7$. By \eqref{eq5:QM1ares}, we obtain  $[\Theta(r),\Theta(r)-r,3]_22$ codes with $r=30,32$, $t=r/2$, $p(\Hb)=2^{m+1}+1$, $\Delta(r,2)=2^{t-5}$.

\textbf{(iii-3)} Let $C_0$ be the $[\Theta(20),\Theta(20)-20,3]_22$ code with $r_0=20$, $t_0=10$, $p(\Hb)=2^6+1$. We have $t_0\ge m\ge \lambda_0+1$  if $m=7,8,9,10$. For these $m$, by \eqref{eq5:QM1ares}, we obtain $[\Theta(r),\Theta(r)-r,3]_22$  codes with $r=34,36,38,40$, $t=r/2$, $p(\Hb)=2^{m+1}+1$, $\Delta(r,2)=2^{t-5}$.

\textbf{(iii-4)}  Let $C_0$ be the $[\Theta(28),\Theta(28)-28,3]_22$ code with $r_0=28$, $t_0=14$, $p(\Hb)=2^6+2$. We have $t_0\ge m\ge \lambda_0+1$ if $m=7,8,\ldots,14$. We take $m=7,8=t_0-7,t_0-6$. By \eqref{eq5:QM1ares}, we obtain  $[\Theta(r),\Theta(r)-r,3]_22$ codes with $r=42,44$, $t=r/2$, $\lambda=\lfloor r/4\rfloor-2=t-13$, $p(\Hb)=2^\lambda+1$, $\lambda+1=t-(\lceil r/4\rceil+1)=t-12$, $\Delta(r,2)=2^{t-5}$.

\textbf{(iii-5)} We consider 8 starting $[\Theta(r_0),\Theta(r_0)-r_0,3]_22$ codes $C_0$ with $r_0=2t_0=30,32,\ldots,44$, $t_0=15,16,\ldots,22$. By \eqref{eq5:QM1ainp}, the allowed values of $m$ satisfy $t_0\ge m\ge \lambda_0+1$, where $\lambda_0$ depends on the specific code. By the above, for these values of $t_0$ the corresponding values of $\lambda_0+1$ are $t_0-7,t_0-7,t_0-8,t_0-8,t_0-8,t_0-8,t_0-12,t_0-12$. For all $C_0$, we use $m=t_0-7,t_0-6$, and by \eqref{eq5:QM1ares}, obtain 16  new $[\Theta(r),\Theta(r)-r,3]_22$ codes with even $r=46,48,\ldots,74,76$,
$t=r/2$, $\lambda=\lfloor r/4\rfloor-2$, $p(\Hb)=2^\lambda+1$, $\Delta(r,2)=2^{r/2-5}$, $\lambda+1=t-(\lceil r/4\rceil+1)\le t-12$, see Table \ref{tab5:LengFun33r64R=2} for $r\le64$.

Since the resulting partition satisfies $p(H)\le2^\lambda+1<2^{\lambda+1}+2$, the hypothesis required for the next application of Theorem \ref{th5:n0Theta} is preserved. Hence the construction can be iterated indefinitely, yielding the infinite code family of \eqref{eq5:InfFamR=2New},  \eqref{eq5:InfFamR=2NewLenFunLambda}.

The upper bound $\overline{\mu}(2)\le2601/2048<1.27002$ directly follows from \eqref{eq1:CovDensity}, \eqref{eq1:liminfdens} and provides the bound $f(2)\le2601/2048<1.27002$.
\end{proof}

\section{The known results on binary linear codes of covering radius 3 and 2-saturating sets}\label{sec6:R=3known}
We introduce the notations:
\begin{align}
&\varphi(r)\triangleq3\cdot2^{(r-1)/3}-1=3\cdot2^{t-1}-1\T{ for }r=3t-2\ge22;\label{eq6:varphi3t-2}\db\\
&\omega(r)\triangleq1024\cdot2^{(r-26)/3}-1=1024\cdot2^{t-9}-1\T{ for }r=3t-1\ge14;\label{eq6:omega3t-1}\db\\
&\upsilon(r)\triangleq822\cdot2^{(r-26)/3}-1=822\cdot2^{t-9}-1\T{ for }r=3t-1\ge26;\label{eq6:gamma3t-1}\db\\
&\psi(r)\triangleq144\cdot2^{(r-18)/3}-1=144\cdot2^{t-6}-1\T{ for }r=3t\ge30\T{ and }r=15;\label{eq6:psi3t}\db\\
&\Upsilon(r)\triangleq 819\cdot2^{(r-26)/3}-1=819\cdot2^{t-9}-1,~ r=26,~r=3t-1\ge44;  \label{eq6:notatUpsR3}\db\\
&\widehat{\Upsilon}(r)\triangleq820\cdot2^{(r-26)/3}-2,~r=38,41. \label{eq6:notathatUpsR3}
\end{align}

\begin{theorem}\label{th6:KnownFamR=3}
There exist infinite families of binary linear $[n,n-r]_2R,\ell$ covering codes with growing codimension $r$, parity check matrix $\Hb$, and parameters as in \eqref{eq6:KnwFamR=3r=3t-2}--\eqref{eq6:KnwFamR=3r=3t-1Ups} and the sporadic covering codes with parameters as in \eqref{eq6:KnwSporadic}, \eqref{eq6:KnwSporadic2}. The codes provide upper bounds on the length functions $\ell_2(r,3)$ and on the smallest sizes $s_2(r-1,2)$ of the $2$-saturating sets in $\PG(r-1,2)$ such that $\ell_2(r,3)=s_2(r-1,2)\le n$, see \eqref{eq1:cov=sat}. The code family \eqref{eq6:KnwFamR=3r=3t-1Ups} provides the upper bound on the constant for Green's Open Problem 40 \cite[Section 5, Problem 40]{Green} $f(3)\le1.36433$.
\begin{description}
  \item[(i)]\cite[Example 4.3,equation (4.12)]{DDL-IEEE}, \cite[Section 5.4, Theorem 5.4.28(ii)]{CHLL-bookCovCod}
\begin{align}
&R=3,~r=3t-2\ge22,~t\ge8,~\ell=2\T{ if }r\ge16,~\overline{\mu}(3)\le2.25,\label{eq6:KnwFamR=3r=3t-2}\db\\
&n=3\cdot2^{(r-1)/3-1}-1=3\cdot2^{t-2}-1=\varphi(r);~\ell_2(r,3)=s_2(r-1,2)\le\varphi(r).\notag
\end{align}

  \item[(ii)] \cite[p.\ 56, Table]{DDL-ACCT2-1990}, \cite[Example 4.1, equation (4.9)]{DDL-IEEE},
  \cite[Section 5.4, Theorem 5.4.28(iii)]{CHLL-bookCovCod}
\begin{align}\label{eq6:KnwFamR=3r=3t-1omega}
&R=3,~r=3t-1\ge14,~t\ge5,~\ell=2,~\overline{\mu}(3)\le2.6667,\db\\
&n=1024\cdot2^{(r-26)/3}-1=1024\cdot2^{t-9}-1=w(r);\,\ell_2(r,3)=s_2(r-1,2)\le w(r).\notag
\end{align}

  \item[(iii)] \cite[p.\ 56, Table]{DDL-ACCT2-1990}, \cite[Example 4.2, equation (4.10)]{DDL-IEEE},
  \cite[Section 5.4, Theorem 5.4.28(iii)]{CHLL-bookCovCod}
\begin{align}\label{eq6:KnwFamR=3r=3t-1ups}
&R=3,~r=3t-1\ge26,~t\ge9,~\ell\ge1,~\overline{\mu}(3)\le1.3794,\db\\
&n=822\cdot2^{(r-26)/3}-2=822\cdot2^{t-9}-2=\upsilon(r);\,\ell_2(r,3)=s_2(r-1,2)\le\upsilon(r).\notag
\end{align}

  \item[(iv)] \cite[Examples 5,\,8]{DavCovNewConstrCovCod}.
\begin{align}\label{eq6:KnwFamR=3r=3t psi}
&R=3,~r=3t\ge30,~t\ge10,\T{ and }r=15,~\ell=2,~\overline{\mu}(3)\le1.8984,
\db\\
&n=144\cdot2^{(r-18)/3}-1=144\cdot2^{t-6}-1=\psi(r);\,\ell_2(r,3)=s_2(r-1,2)\le\psi(r).\notag
\end{align}

  \item[(v)]\cite[Theorems 7.6, 7.7, equations (7.7), (7.10), (7.11)]{DMP-arXiv2025},\cite[Theorems 18, 19, equations (57), (60), (61)]{DMP-IEEE2026}
\begin{align}\label{eq6:KnwFamR=3r=3t-1Ups}
&R=3,~r=3t-1\ge44\T{ and }r=26,~t=9\T{ and }t\ge15,~\ell=2, \db\\
&n=819\cdot2^{(r-26)/3}-1=819\cdot2^{t-9}-1=\Upsilon(r),~\overline{\mu}(3)\le1.36433;~\notag\db\\
&p(\Hb,0)=p(\Hb,1)=35,~p(\Hb,2)=818\T{ if }r=26;~p(\Hb,1)=2^{(r-23)/3} \notag\db\\
&\T{if }r=44,47,50,53;~p(\Hb,2)=819\T{ if }r\ge56;~\ell_2(r,3)=s_2(r-1,2)\le\Upsilon(r).
\notag
\end{align}

  \item[(vi)]\cite[Theorem 7.6, equations (7.8), (7.9)]{DMP-arXiv2025},\cite[Theorem 18, equations (58), (59)]{DMP-IEEE2026}
\begin{align}\label{eq6:KnwSporadic}
&R=3,~r=38,41,~\mu(38)\le1.36871,~\mu(41)\le1.36902,~\ell=1,\db\\
& n=820\cdot2^{(r-26)/3}-2=\widehat{\Upsilon}(r);~p(\Hb,0)=57,p(\Hb,1)=59\T{ if }r=38;\notag\db\\
&p(\Hb,0)=89,p(\Hb,1)=91\T{ if }r=41;~\ell_2(r,3)=s_2(r-1,2)\le\widehat{\Upsilon}(r).\notag
\end{align}

\item[(vii)] \cite[Example 4.4, equation (4.13)]{DDL-IEEE}
\begin{align}\label{eq6:KnwSporadic2}
&R=3,~r=18,~n=153,~\mu(18)\le2.27760,~\ell\ge1;\db\\
&\ell_2(18,3)=s_2(17,2)\le153.\notag
\end{align}
\end{description}
\end{theorem}

\begin{theorem}\label{th6:DDLD95R=3} \cite[Theorem 4.1, equations (4.3), (4.4), (4.7), (4.8)]{DDL-IEEE}
\textbf{Constructions} $\mathbf{\QM_1^3,~\QM_2^3}$. In the Basic Construction \emph{QM} of Section $\ref{sec3:qm concat}$, let the starting code $C_0$  be an $[n_0,n_0-r_0]_23,\ell_0$ code
with a parity check matrix $\Hb_0$ of the form \eqref{eq32:H0}. We define a new $[n,n-r]_2R_C,\ell$ code $C$ by a parity check matrix $\Hb_C$ of the form \eqref{eq32:QM-H}--\eqref{eq32:Wbm}, where the indicator set $\Bs$, parameters $\ell_0,m$, and submatrix $\D$ are as follows:
\begin{align}
&\QM_1^3: ~\ell_0=2,~\Bs\subseteq\F_{2^m}\cup\{*\},~2^m+1\ge p(\Hb_0,2),~
 \D=\D_2\triangleq\left[\begin{array}{c}
 \zb_{r_0+m}\\
 \Wb_m\\
 \zb_{m}
 \end{array}\right].\label{eq6:ell0=2inp}\db\\
&\QM_2^3:~\forall\ell_0,~\Bs=\F_{2^m}\cup\{*\},~n_0\ge2^m+1\ge p(\Hb_0,\ell_0),~
 \D=\D_2.\label{eq6:any ell0CSIinp}
\end{align}
Then the new code $C$ is an $[n, n-r]_2R,\ell$ code with parameters:
\begin{align}
&\QM_1^3:~R=3,~n=2^m(n_0+1)-1,~r=r_0+3m,~\ell=2,~\label{eq6:ell0=2res}\db\\
&p(\Hb_C,2)\le p(\Hb_0,2)+1.\notag\db\\
&\QM_2^3:~R=3, ~n=2^m(n_0+1)-1,~r=r_0+3m,~\ell=2\T{ for }m\ge2,\label{eq6:any ell0CSIres}\db\\
&p(\Hb_C,\ell_0)\le2^{m+1}+3.\notag
\end{align}
\end{theorem}

\begin{theorem}\label{th6:OK189}\"{O}sterg{\aa}rd, Kaikkonen \cite[Table 2]{OstKaikUpBndBin1998}
Let $C_{OK}$ be the binary $[18,9]_2$ code with the parity check matrix
$\Hb_{OK}=[\Ib_{9}\Mb_{OK}]$, where $\Mb_{OK}$ is the $9\times9$ binary matrix with the following columns in
hexadecimal notation:
\begin{align}\label{eq6:HOK}
&\Mb_{OK}=[\mathrm{1A0,174,A5,173,17,E8,9,18D,1CE}].
\end{align}
Then $C_{OK}$ is a $[18,9]_23$ code of covering radius $3$.
\end{theorem}

\begin{theorem}\label{th6:PartitionOK}\cite[Theorem 7.1]{DMP-arXiv2025},\cite[Theorem 15]{DMP-IEEE2026}
  Let $\Hb_{OK}=[\Ib_9\Mb_{OK}]$ be the binary $9\times18$ matrix with columns of the $9\times9$ submatrix $\Mb_{OK}$ given in \eqref{eq6:HOK}. Let the columns of $\Hb_{OK}$ be numbered from left to right. Let the partition $\Ps_{OK}$ of the column set of $\Hb_{OK}$ into $11$ subsets  be as follows:
  \begin{equation*}
 \Ps_{OK}\triangleq\{1,2,4\},\{3\},\{5,8\},\{6,17\},\{7,10\},\{11,14\},\{12\},\{13,18\},\{15\},\{9\},\{16\},\notag
  \end{equation*}
where the column numbers in each subset are listed. Then, $\Ps_{OK}$ is a $(3,1)$-partition of $\Hb_{OK}$,
 i.e. $p(\Hb_{OK},1;\Ps_{OK})=11$ and
the $[18,9]_23$ code $C_{OK}$ with the parity check matrix $\Hb_{OK}$ of Theorem \ref{th6:OK189} is an $[18,9,3]_23,1$ code with $p(\Hb_{OK},1)\le11$.
\end{theorem}

\begin{theorem}\label{th6:OK3725}\"{O}sterg{\aa}rd, Kaikkonen \cite[Table 2]{OstKaikUpBndBin1998}
Let $C_{OK2}$ be the binary $[37,25]_2$ code with the parity check matrix
$\Hb_{OK2}=[\Ib_{12}\Mb_{OK2}]$, where $\Mb_{OK2}$ is the $12\times25$ binary matrix with the following columns in
hexadecimal notation:
\begin{align}\label{eq6:HOKoverl}
&\Mb_{OK2}=[\mathrm{BFF,727,977,93B,6B9,51D,AAE,9B8,85B,A2D,B16,9C5,AE2,}\db\\
&\mathrm{6E2,371,5B8,2DC,16E,B7,45B,62D,716,38B,5C5,7FF}].\notag
\end{align}
Then $C_{OK2}$ is a $[37,25]_23$ code of covering radius $3$.
\end{theorem}

\section{New results on binary linear codes of covering radius 3 and 2-saturating sets}\label{sec7:R=3new}

\begin{lemma}\label{lem7:PartitionOK}
Let $C_{OK}$ be the binary $[18,9,3]_23$ code with the parity check $9\times18$ matrix
$\Hb_{OK}=[\Ib_{9}\Mb_{OK}]$, where $\Mb_{OK}$ is the $9\times9$ binary matrix
 with columns given in \eqref{eq6:HOK}, see Theorems \ref{th6:OK189} and \ref{th6:PartitionOK}. Let the columns of $\Hb_{OK}$ be numbered from left to right.
 \begin{description}
   \item[(i)] Let the partition $\widehat{\Ps}_{OK}$ of the column set of $\Hb_{OK}$ into $9$ subsets be as follows:
  \begin{equation*}
 \widehat{\Ps}_{OK}\triangleq \{1\}, \{3\}, \{8\}, \{12\}, \{7,15\}, \{14,18\}, \{2,4,10\}, \{5,11,13\}, \{6,9,16,17\}.
 \end{equation*}
where the column numbers in each subset are listed. Then, $\widehat{\Ps}_{OK}$ is a $(3,0)$-par-\\tition of $\Hb_{OK}$, i.e. $p(\Hb_{OK},0;\widehat{\Ps}_{OK})=9$.
   \item[(ii)]  The code $C_{OK}$ is an $[18,9,3]_23,1$ code, however it is not an $[18,9,3]_23,2$ code.
 \end{description}
\end{lemma}

\begin{proof}
\begin{description}
   \item[(i)] We obtained $\widehat{\Ps}_{OK}$ by computer search.
   \item[(ii)] By Theorem \ref{th6:PartitionOK}, $C_{OK}$ is an $[18,9,3]_23,1$ code. Using an exhaustive computer search, we showed that seven columns of the matrix $\Hb_{OK}$ (with numbers 3, 7, 8, 12, 15, 17, 18) cannot be represented as a sum of two or three distinct columns of it. Therefore, $C_{OK}$ is not an $[18,9,3]_23,2$ code.
 \end{description}
\end{proof}
In Theorems \ref{th7:151R=3}, \ref{th7:r=26,38}, \ref{th7:newr=3t-1a}, let the covering density $\mu(r,3)$ and the corresponding asymptotic covering density $\overline{\mu}(3)$ be as in $\eqref{eq1:CovDensity}$ and \eqref{eq1:liminfdens}. New $[n,n-r]_23$ codes, obtained in these theorems, provide upper bounds on the length functions $\ell_2(r,3)$ and on the smallest sizes $s_2(r-1,2)$ of the $2$-saturating sets in $\PG(r-1,2)$ such that $\ell_2(r,3)=s_2(r-1,2)\le n$, see \eqref{eq1:cov=sat}. The decrease of the known upper bounds on the length function $\ell_2(r,3)$ is $\Delta(r,3)$.

\begin{theorem}\label{th7:151R=3}
There exists a new $[151,133,3]_23,2$ code with a parity check matrix, admitting a $(3,0)$-partition into $19$ subsets; its parameters and bounds are as follows:
\begin{align}\label{eq7:151R=3}
&R=3,~r=18,~n=151,~\mu(18)\le2.18946,~\ell=2;\db\\
&\ell_2(18,3)=s_2(17,2)\le151;~ \Delta(18,3)=2.\notag
\end{align}
\end{theorem}

\begin{proof}
 We use Construction $\QM_2^3$ of Theorem \ref{th6:DDLD95R=3}. As the code $C_0$ we take the $[18,9,3]_23$ code $C_{OK}$ with the (3,0)-partition $\widehat{\Ps}_{OK}$, see Theorems \ref{th6:OK189}, \ref{th6:PartitionOK}, and Lemma \ref{lem7:PartitionOK}. We have $n_0=18>2^3+1=p(\Hb_0,0)=9$ that allows us to take $m=3$ in \eqref{eq6:any ell0CSIinp}. By \eqref{eq6:any ell0CSIres}, we obtain a new $[n,n-18]_23,2$ code $C$ with $n=151$ and $p(\Hb_C,0)\le2^{m+1}+3$. By Theorem \ref{th6:KnownFamR=3}(vii) with \eqref{eq6:KnwSporadic2}, the known value is $n=153$.
\end{proof}

\begin{lemma}\label{th7:PartitionOKoverl18}
 Let $C_{OK2}$ be the $[37,25]_23$ code of Theorem \ref{th6:OK3725}. Let $\Hb_{OK2}=[\Ib_{12}\Mb_{OK2}]$ be its $12\times37$ parity check matrix with columns of the $12\times25$ submatrix $\Mb_{OK2}$ given in \eqref{eq6:HOKoverl}. Let the columns of $\Hb_{OK2}$ be numbered from left to right. Let the partition $\Ps_{OK2}$ of the column set of $\Hb_{OK2}$ into $18$ subsets  be as follows:
  \begin{align}\label{eq7:PartitionOKoverl18}
 &\Ps_{OK2}\triangleq\{3\},\{4\},\{5\},\{6\},\{7\},\{8\},\{9\},\{10\},\{11\},\{12\},\{27\},\{29\},\{30\},\{31\},\db\\
 &\{35\},\{1,2,16,20,21,23,34\},\{14,17,18,24,26,32,36,37\},\{13,15,19,22,25,28,33\},\notag
  \end{align}
where the column numbers in each subset are listed. Then, $\Ps_{OK2}$ is a $(3,0)$-partition of $\Hb_{OK2}$,
 i.e. $p(\Hb_{OK2},0;\Ps_{OK2})=18$. Also, the minimum distance of $C_{OK2}$ is 4. So,  $C_{OK2}$ is a $[37,25,4]_23,0$ code with $p(\Hb_{OK2},0)\le18$.
\end{lemma}

\begin{proof}
 The assertions were verified exhaustively by computer search: all $2^{12}$ vectors of $\F_2^{12}$ were checked for representation by at most three columns from distinct subsets of $\Ps_{OK2}$, and the minimum distance was verified independently.
\end{proof}

\begin{lemma}\label{lem7:propCWu}
We consider the parity check matrix $\Hb_{Wu}$ \eqref{eq5:Hwu} of the $[50,40,3]_22$ code $C_{Wu}$, see Theorem \ref{th5:Wu code}. The following holds:
\begin{description}
             \item[(1)] Exactly $13$ columns of the matrix $\Hb_{Wu}$ ($32,329,341,373,479,491,576,608,734,777,$
             $789,821,927$)
             are the sum of two other distinct columns; all  the columns of $\F_2^{10}\setminus\Hb_{Wu}$ (apart from the zero column) can be obtained as the sum of two distinct columns of $\Hb_{Wu}$.

             \item[(2)] All the $50$ columns of the matrix $\Hb_{Wu}$ can be obtained as the sum of three other distinct its columns.

             \item[(3)] Any column of $F_2^{10}$ (including the zero column) can be obtained as the sum of three distinct columns of $\Hb_{Wu}$.
           \end{description}
\end{lemma}

\begin{proof}
 The assertions were verified exhaustively by computer search using the system MAGMA \cite{Magma}.
\end{proof}

\begin{lemma}\label{lem7:3subsets}
Let $m\ge3$. Let the field $\F_{2^m}=\{\xi_1,\xi_2,\ldots ,\xi_{2^m}\}$, $\xi_1=0,\xi_2=1$, see \eqref{eq32:xi}. Let a field element $\xi_j$ be represented by the binary $m$-tuple corresponding to the integer $j-1$; in particular, $\xi_1=\underbrace{0\ldots0}_m$,
$\xi_2=\underbrace{0\ldots0}_{m-1}1$, $\xi_3=\underbrace{0\ldots0}_{m-2}10$, $\xi_{2^m}=\underbrace{1\ldots1}_m$.
Let the partitions into three subsets $\Ps_{2^m}$ of the field $\F_{2^m}$ and $\Ps_{2^m}^*$ of the set $\F_{2^m}^*=\F_{2^m}\backslash\{0\}=\F_{2^m}\backslash\{\xi_1\}$ be as follows.
\begin{align}
&\Ps_{2^m}=\{\xi_1,\xi_2\}, \{\xi_{2t-1}\,|\, t=2,3,\ldots,2^{m-1}\}, \{\xi_{2t}\,|\, t=2,3,\ldots,2^{m-1}\}; \label{eq7:3subsetsF}\\
&\Ps_{2^m}^*=\{\xi_2\}, \{\xi_{2t-1}\,|\, t=2,3,\ldots,2^{m-1}\}, \{\xi_{2t}\,|\, t=2,3,\ldots,2^{m-1}\}.\label{eq7:3subsetsF*}
\end{align}
Then every element of $\F_{2^m}$ (resp. $\F_{2^m}^*$) can be represented as the sum of three (resp. two) elements from distinct subsets of $\Ps_{2^m}$ (resp. $\Ps_{2^m}^*$).
\end{lemma}

\begin{proof}
For $\F_{2^m}$, we have $\xi_1=\xi_2+\xi_{2t-1}+\xi_{2t}$, $\xi_2=\xi_1+\xi_{2t-1}+\xi_{2t}$. Let $S\triangleq\{\xi_{2u-1}+\xi_{2v}\,|\, u=2,3,\ldots,2^{m-1}, v=2,3,\ldots,2^{m-1}\}$. Every element $\xi_{2t}$ corresponds to an odd binary integer. Any odd integer in $\{1,\ldots,2^m-1\}$ can be written as the bitwise XOR of an even integer $2u-2$ and an odd integer $2v-1$, with $u,v\in\{2,\ldots,2^{m-1}\}$. Hence
$S=\{\xi_{2t}\,|\, t=1,2,\ldots,2^{m-1}\}$. For $t\ge2$,  $\xi_{2t-1}=\xi_2+\xi_{2u-1}+\xi_{2v}$, $\xi_{2t}=\xi_1+\xi_{2u-1}+\xi_{2v}$, for a suitable pair $u,v\in\{2,\ldots,2^{m-}\}$.

For $\F_{2^m}^*$ and $t\ge2$, we have $\xi_2=\xi_{2t-1}+\xi_{2t}$, $\xi_{2t-1}=\xi_{2t}+\xi_2$, $\xi_{2t}=\xi_{2t-1}+\xi_2$.
\end{proof}

\begin{theorem}\label{th7:ell0=0R=3}
$\mathbf{Construction~\QM_3^3}$. In the Basic Construction \emph{QM} of Section $\ref{sec3:qm concat}$, let the starting code $C_0$  be an $[n_0,n_0-r_0]_23,0$ code
with a parity check matrix $\Hb_0$ of the form \eqref{eq32:H0} admitting a $(3,0)$-partition $\Ps_0$
into $p(\Hb_0,0)$ subsets. Let $\Vc_{2m}$ be an $[n_{2m},n_{2m}-2m]_22$  code with a parity check matrix $\Hcb_{2m}$ admitting a $2$-partition $\Ps_{2m}$ into $p(\Hcb_{2m})$ subsets.
We define a new $[n,n-r]_2R_C,\ell_C$ code $C$ by a parity check matrix $\Hb_C$ of the form \eqref{eq32:QM-H}, \eqref{eq32:Aj}, \eqref{eq32:xi}--\eqref{eq32:Hcb2m}, where the indicator set $\Bs$, parameter $m$, and submatrix $\D$ are as follows:
\begin{equation}
\QM_3^3:~\Bs\subseteq\F_{2^m}^*,~2^m-1\ge p(\Hb_0,0),~m\ge2,~
 \D=\left[\begin{array}{@{}cc@{}}
 \zb_{r_0}&\zb_{r_0}\\
 \Wb_m&\zb_{m}\\
 \zb_{2m}&\Hcb_{2m}
 \end{array}\right].\label{eq7:ell0=0inpR=3}
\end{equation}
Then the new code $C$ is an $[n, n - r]_23,\ell_C$ code of covering radius $3$ with the parameters:
\begin{align}
&\QM_3^3:~n=2^m(n_0+1)+n_{2m}-1,~r = r_0 + 3m,~\ell_C\ge1;\label{eq7:ell0=0resR=3}\db\\
&\T{if any vector of $\F_2^{2m}$ is a sum of two or three distinct columns of $\Hcb_{2m}$ then }\ell_C=2;\label{eq7:part-ellC=2}\db\\
&p(\Hb_C,0)\le p(\Hb_0,0)+p(\Hcb_{2m})+1,~p(\Hb_C,1)\le p(\Hb_0,0)+p(\Hcb_{2m})+3;\label{eq7:part-ellR=3}\db\\
&p(\Hb_C,2)\le 3p(\Hb_0,0)+n_{2m}+3\T{ if }m\ge3.\notag
\end{align}
If there exists a triple of linearly dependent columns of $\Hb_0$ or $\Hcb_{2m}$  belonging to distinct subsets of the partitions $\Ps_0$ or $\Ps_{2m}$, respectively, then $p(\Hb_C,1)= p(\Hb_C,0)$.
\end{theorem}

\begin{proof}
We prove the assertions considering the representation as a sum of $\le3$ columns of $\Hb_C$ for an arbitrary vector $\Ub=(\pib,\ub_1,\ub_2,\ub_3)^{tr}\in\F_2^{\,r}$, where $\pib\in\F_2^{\,r_0}$, $\ub_1,\ub_2,\ub_3\in\F_2^{\,m}$.

\begin{description}
  \item[(i)] $\pib=\hb_i+\hb_j+\hb_k$; $\hb_i$, $\hb_j$, $\hb_k$ belong to distinct subsets of $\Ps_0$; $\beta_i,\beta_j,\beta_k$ are distinct.

  We solve the linear system $\beta_i^{\delta-1} \xb+\beta_j^{\delta-1} \yb+\beta_k^{\delta-1} \zbl=\ub_\delta$, $\delta=1,2,3$, for $\xb,\yb,\zbl\in\F_{2^m}$. The coefficient matrix is a $3\times3$  Vandermonde matrix with distinct nonzero parameters $\beta_i,\beta_j,\beta_k$, hence it is nonsingular over $\F_{2^m}$, and the system has a unique solution.
  Now $\Ub= (\hb_i,\xb,\beta_i \xb,\beta_i^2 \xb)^{tr} +(\hb_j,\yb,\beta_j \yb,\beta_j^2 \yb)^{tr}+(\hb_k,\zbl,\beta_k \zbl,\beta_k^2 \zbl)^{tr}$.

  \item[(ii)] $\pib=\hb_i+\hb_j$; $\hb_i$ and $\hb_j$ belong to distinct subsets of $\Ps_0$; $\beta_i\ne\beta_j$, $\beta_i,\beta_j\ne0$.

  We solve the linear system $\beta_i^{\delta-1} \xb+\beta_j^{\delta-1} \yb=\ub_\delta$, $\delta=2,3$, for $\xb,\yb\in\F_{2^m}$.
  Since $\beta_i\ne\beta_j$, $\beta_i,\beta_j\ne0$, the coefficient matrix is nonsingular over $\F_{2^m}$, and the system has a unique solution. Now $\Ub=(\hb_i,\xb,\beta_i \xb,\beta_i^2 \xb)^{tr} +(\hb_j,\yb,\beta_j \yb,\beta_j^2 \yb)^{tr}+(\zb_{r_0,}\wb,\zb_{2m})^{tr}$ where $\wb=\xb+\yb+\ub_1\in\Wb_m$. If $\xb+\yb=\ub_1$ then  $(\zb_{r_0,}\wb,\zb_{2m})^{tr}$ is absent.

\item[(iii)] $\pib=\hb_j$.

  Let $(\beta_j\ub_1,\beta_j^2\ub_1)^{tr}\ne(\ub_2,\ub_3)^{tr}$. Then $\Ub=(\hb_j,\ub_1,\beta_j\ub_1,\beta_j^2\ub_1)^{tr} +\boldsymbol{S}$, where $\boldsymbol{S}$ is a single column or a sum of two columns of $\Hcb_{2m}$ belonging to distinct subsets of $\Ps_{2m}$, since $\Vc_{2m}$ is a code with covering radius two.

Let $(\beta_j\ub_1,\beta_j^2\ub_1)^{tr}=(\ub_2,\ub_3)^{tr}$. For the new code $C$, we consider $\ell_C\in\{0,1,2\}$. For $\ell_C\in\{0,1\}$, it is sufficient to put $\Ub=(\hb_j,\ub_1,\beta_j\ub_1,\beta_j^2\ub_1)^{tr}$.
  To provide $\ell_C=2$, by Lemma \ref{lem7:3subsets}, we partition the column set of every matrix $\Ab(\hb_j,\beta_j)$ \eqref{eq32:Aj} into three subsets in accordance to \eqref{eq7:3subsetsF}. This allows us to represent $\ub_1=\xi_{i_1}+\xi_{i_2}+\xi_{i_3}$ where $\xi_{i_1},\xi_{i_2},\xi_{i_3}$ belong to distinct subsets of $\Ps_{2^m}$.
  Now  $\Ub=\sum_{u=1}^3(\hb_j,\xi_{i_u},\beta_j\xi_{i_u},\beta_j^2\xi_{i_u})^{tr}$.

\item[(iv)] $\pib=\zb_{r_0}$.

  The last $3m$ rows of $\D$ form a parity check matrix of the direct sum \cite{GrahSlo} of the $[2^m-1,2^m-1-m,3]_21$ Hamming code and the code $\Vc_{2m}$ that gives a code of covering radius 3. We
consider the representation of $\Ub$ as a sum of columns of $\D$; this helps us to estimate the cardinality of  the $(3,\ell)$-partitions of $\Hb_C$.

Let $\ub_1=\wb\in\Wb_m$.

If $(\ub_2, \ub_3)^{tr}\ne(\zb_{2m})^{tr}$, then $\Ub=(\zb_{r_0},\wb,\zb_{2m})^{tr} +\boldsymbol{S}$, where $\boldsymbol{S}$ is a single column or a sum of two columns of $\Hcb_{2m}$ belonging to distinct subsets of $\Ps_{2m}$.

If $(\ub_2, \ub_3)^{tr}=(\zb_{2m})^{tr}$, then for $\ell_C\in\{0,1\}$, it is sufficient to put $\Ub=(\zb_{r_0},\wb,\zb_{2m})^{tr} $.
  To provide $\ell_C=2$, by Lemma \ref{lem7:3subsets}, we partition the column set of $\Wb_m$ \eqref{eq32:Wbm} into three subsets in accordance to \eqref{eq7:3subsetsF*}.  This allows us to represent $\wb=\wb_1+\wb_2$ where $\wb_1,\wb_2$belong to distinct subsets of $\Ps_{2^m}^*$. Now $\Ub=\sum_{i=1}^2(\zb_{r_0},\wb_i,\zb_{2m})^{tr}$.

Let $\ub_1=\zb_m$.

If $(\ub_2, \ub_3)^{tr}=0$, then $\Ub=(\pib,\ub_1,\ub_2,\ub_3)^{tr}$ is the zero vector of $\F_2^{\,r}$. For an $(3,0)$-partition we can formally treat $\Ub$ as the sum of 0 columns, see Definition \ref{def31:partition}(i). This implies, $p(\Hb_C,0)\le p(\Hb_0,0)+p(\Hcb_{2m})+1$ in \eqref{eq7:part-ellR=3}. To provide $\ell_C=1$, by Lemma~\ref{lem7:3subsets}, we partition the column set of $\Wb_m$ \eqref{eq32:Wbm} into three subsets in accordance to \eqref{eq7:3subsetsF*}. This implies, $p(\Hb_C,1)\le p(\Hb_0,0)+p(\Hcb_{2m})+3$ in \eqref{eq7:part-ellR=3}.

If $(\ub_2, \ub_3)^{tr}\ne0$ and it is not a column of $\Hcb_{2m}$, then $(\ub_2, \ub_3)^{tr}$ is a sum of two columns of $\Hcb_{2m}$.

If $(\ub_2, \ub_3)^{tr}$ is a column, say $\boldsymbol{\lambda}_{2m}$, of $\Hcb_{2m}$, then for $\ell_C\in\{0,1\}$, it is sufficient to put $\Ub=(\zb_{r_0+m},\boldsymbol{\lambda}_{2m})^{tr} $. To provide $\ell_C=2$, we
should represent $\boldsymbol{\lambda}_{2m}$ as a sum of two or three columns of $\Hcb_{2m}$. Here we can use the trivial partition of $\Hcb_{2m}$.
\end{description}

Thus, we have proved that any vector $\Ub \in\F_2^{\,r}$  can be represented as  a sum of at most three columns of $\Hb_C$. This means that the covering radius of the new code $C$ is $R_C=3$.

To estimate $p(\Hb_C,0)$ for the matrix $\Hb_C$ of the form \eqref{eq32:QM-H}, \eqref{eq32:Aj}, \eqref{eq32:xi}--\eqref{eq32:Hcb2m}, \eqref{eq7:ell0=0inpR=3},
we partition the columns of $\Hb_C$ (except $\D$) in accordance with the partition
$\Ps_0$ of $\Hb_0$: if columns $\hb_i, \hb_j$ of $\Hb_0$ belong to distinct subsets of
$\Ps_0$, then the columns of submatrices $\Ab(\hb_i,\beta_i),\Ab(\hb_j,\beta_j)$ belong to distinct subsets
of the partition of $\Hb_C$. Also, taking into account the proof above, we partition the columns of $\D$ \eqref{eq7:ell0=0inpR=3} into $p(\Hcb_{2m})+1$ subsets so that the submatrix $\Wb_m$ forms one subset. This implies $p(\Hb_C,0)\le p(\Hb_0,0)+p(\Hcb_{2m})+1$ in \eqref{eq7:part-ellR=3}.

 Now we note that, since the Hamming code has distance three, $\Ub=(\zb_{r_0+3m})^{tr}$ can always be represented as a sum of three columns of $\Wb_m$; this is necessary for $\ell_C\ge1$, see Section \ref{sec3:qm concat}. To provide this, we should partition $\Wb_m$ into three subsets according to Lemma \ref{lem7:3subsets} with \eqref{eq7:3subsetsF*}, that explains the inequality  $p(\Hb_C,1)\le p(\Hb_C,0)+2$ in~\eqref{eq7:part-ellR=3}. We could avoid the above partition, if there exists a triple of linearly dependent columns of $\Hb_0$ or $\Hcb_{2m}$ belonging to distinct subsets of the corresponding partitions $\Ps_0$ or $\Ps_{2m}$; this explains the case $p(\Hb_C,1)= p(\Hb_C,0)$.

 Finally, by the proof above, if any vector of $\F_2^{2m}$ is equal to a sum of two or three columns of $\Hcb_{2m}$ then $\ell_C=2$. In this case we can use the trivial partition of $\Hcb_{2m}$ and put $p(\Hb_C,2)\le 3p(\Hb_0,0)+n_{2m}+3\T{ if }m\ge3$, see the cases (iii), (iv).
 \end{proof}

\begin{remark}
  In comparison with Construction $QM_5^3$ \cite[Theorem 7.3]{DMP-arXiv2025}, \cite[Theorem~17]{DMP-IEEE2026},
  Construction $\QM_3^3$ of Theorem \ref{th7:ell0=0R=3} significantly reduces (i.e. improves) the upper bound on $p(\Hb_C,2)$: we obtained $p(\Hb_C,2)\le 3p(\Hb_0,0)+n_{2m}+3$ instead of $p(\Hb_C,2)\le n$. For example, in Theorem \ref{th7:r=26,38}, for $r=26$, where $n=817$, now we have $p(\Hb_C,2)\le122$ whereas the previous bound provides $p(\Hb_C,2)\le817$. The aforementioned improvement significantly expands the application range of Construction $\QM_3^3$; this gives new results, better than known ones. The improvement was made possible by the new Lemma \ref{lem7:3subsets}. The lemma is used in a few places of the proof of Theorem \ref{th7:ell0=0R=3}.
  \end{remark}

We introduce the notations.
 \begin{align}
&\Omega(r)=\Omega(3t-1)\triangleq 818\cdot2^{(r-26)/3}-1=818\cdot2^{t-9}-1,~ r=26,~r=3t-1\ge41.\label{eq7:notatOmeR3}\db\\
&\widehat{\Omega}(r)\triangleq819\cdot2^{(r-26)/3}-2,~r=38.
  \label{eq7:notatOmehatR3}
 \end{align}

 \begin{theorem}\label{th7:r=26,38}
There are two new $[n,n-r,3]_23,\ell$ codes $C$ with a parity check matrix $\Hb_C$ of the form \eqref{eq32:QM-H}, \eqref{eq32:Aj}, \eqref{eq32:xi}--\eqref{eq32:Hcb2m}, \eqref{eq7:ell0=0inpR=3} and parameters and bounds as in \eqref{eq7:r=26}, \eqref{eq7:r=38}.
\begin{align}
&r=26,~n=817=\Omega(26),~\mu\le1.35438,~\ell=2,~p(\Hb_C,0)=p(\Hb_C,1)=33,\label{eq7:r=26}\db\\
&\phantom{r=26,~}p(\Hb_C,2)\le122;~\ell_2(26,3)=s_2(25,2)\le817=\Omega(26);~\Delta(26,3)=1.\notag\db\\
&r=38,~n=13102=\widehat{\Omega}(38),~\mu\le1.36371,~\ell=1,~p(\Hb_C,0)\le57,\label{eq7:r=38}\db\\
&\phantom{r=38,~}p(\Hb_C,1)\le59;~\ell_2(38,3)=s_2(37,2)\le13102=\widehat{\Omega}(38);~\Delta(38,3)=2^4.\notag
\end{align}
\end{theorem}

\begin{proof}
  We apply Construction $\QM_3^3$ of Theorem \ref{th7:ell0=0R=3}.  Let the starting code $C_0$  be the perfect $[23,12,7]_23,0$ Golay code \cite{Golay1949} with the parity check matrix $\Hb_0$; we use the trivial $(3,0)$-partition of it. So, $n_0=23$, $r_0=11$, $\ell_0=0$, $p(\Hb_0,0)=23$.

  By \eqref{eq7:ell0=0inpR=3}, we can put $m=5,9$ that provides $2^m-1\ge23$ and implies $r=26,38$, respectively.
  For $m=5,9$, respectively, we take $\Vc_{2m}$ to be the $[50,40,3]_22$ and $[815,797,3]_22$ codes supplied by Theorems \ref{th5:Wu code}, \ref{th5:r=18}, Lemma \ref{lem7:propCWu}, and Table \ref{tab5:LengFun r32R=2}. Their partition sizes are $p(\Hcb_{10})=10$, $p(\Hcb_{18})=2^5+1$.

  By \eqref{eq7:ell0=0resR=3}, \eqref{eq7:notatOmeR3}, \eqref{eq7:notatOmehatR3}, we obtain $n=817=\Omega(26)$, if $r=26$, and $n=\widehat{\Omega}(r)$ if $r=38$. By \eqref{eq7:ell0=0resR=3}, $\ell=1$  if $r=38$.
  For $r=26$, Lemma \ref{lem7:propCWu} implies that every vector of $\F_2^{10}$ is a sum of three distinct columns of $\Hcb_{10}$; hence \eqref{eq7:part-ellC=2} gives $\ell=2$; also the values $\ell=0$, $\ell=1$ can be considered, see Definition \ref{def31:partition}(ii).

  By \eqref{eq7:part-ellR=3}, we obtain $p(\Hb_C,0)\le57$, $p(\Hb_C,1)\le59$ for $r=38$.

For $r=26$, the constructive $(3,0)$-partition from the proof of Theorem \ref{th7:ell0=0R=3} has size $34=23+10+1$ (23 subsets of the form $\Ab(\hb_i,\beta_i)$ \eqref{eq32:Aj}, $i=1,2,\ldots,23$, where $\hb_i$ is the $i$-th column of a parity check matrix of the $[23,12,7]_23,0$ Golay code; the ten subsets of $\Ps_{Wu}$ \eqref{eq5:PWu}; and one subset $\Wb_5$). An exhaustive computer check of all $\binom{34}{2}=561$ unions of two of those subsets finds exactly two that remain a $(3,0)$-partition; one of them is the merge we use: merge one subset of the form $\Ab(\hb_i,\beta_i)$ with the singleton subset $\{403\}$ of $\Ps_{Wu}$. The resulting $(3,0)$-partition of 33 subsets covers every vector of $\F_2^{26}$ by at most three columns from distinct subsets. Moreover, it is also a $(3,1)$-partition since in $\Ps_{Wu}$ there is the triple of linearly dependent columns $491, 734, 821$ belonging to distinct subsets, see Theorem \ref{th5:Wu code}(iii).

Furthermore, for $r=26$, where $m=5>3$, we have  $p(\Hb_C,2)\le3\cdot23+50+3=122$, see \eqref{eq7:part-ellR=3}.

  The decrease $\Delta(r,3)$ follows from $\ell_2(r,3)\le n$ and Theorem \ref{th6:KnownFamR=3}.
\end{proof}

\begin{theorem}\label{th7:newr=3t-1a}
Let $\Omega(r)$ be as in \eqref{eq7:notatOmeR3}. There exists a new infinite family of binary linear  $[n,n-r,3]_23,2$ codes with a parity check matrix $\Hb$ of the form \eqref{eq32:QM-H}--\eqref{eq32:Hcb2m}, and parameters and bounds as in \eqref{eq7:InfFamR=3New}.
\begin{align}\label{eq7:InfFamR=3New}
&R=3,~r=3t-1\ge41\T{ and }r=26,~t=9\T{ and }t\ge14,~\ell=2, \db\\
&n=818\cdot2^{(r-26)/3}-1=818\cdot2^{t-9}-1=\Omega(r),~\overline{\mu}(3)
\le1.35935;\notag\db\\
&p(\Hb,0)=p(\Hb,1)=33\T{ and }p(\Hb,2)\le122\T{ if }r=26;\notag\db\\
&p(\Hb,0)=p(\Hb,1)\le67\T{ if }r=41,~p(\Hb,0)=p(\Hb,1)\le131\T{ if }r=44;\notag\db\\
&~p(\Hb,2)\le123\T{ if }r\ge47;~\ell_2(r,3)=\ell_2(3t-1,3)=s_2(3t-2,2)\le818\cdot2^{t-9}-1=\Omega(r),\notag\db\\
&\Delta(41,3)=63,~\Delta(r,3)=2^{(r-26)/3}\T{ if }r\ge44.\notag
\end{align}
\end{theorem}

\begin{proof} For $r=26$, see Theorem \ref{th7:r=26,38}.

Then we use Construction $\QM_2^3$ \eqref{eq6:any ell0CSIinp}, \eqref{eq6:any ell0CSIres}, of Theorem \ref{th6:DDLD95R=3}. As $C_0$ we take the $[\Omega(26)=817, \Omega(26)-26,3]_23,\ell$ code of Theorem \ref{th7:r=26,38}. We consider $\ell=0,1$ with
$p(\Hb_0,0)=p(\Hb_0,1)=33$ that allows us, by \eqref{eq6:any ell0CSIinp}, to take $m=5,6$ and obtain the needed codes for $r=41,44$.

Finally, we use Construction $\QM_1^3$ \eqref{eq6:ell0=2inp}, \eqref{eq6:ell0=2res}, of Theorem \ref{th6:DDLD95R=3}. Again we take $C_0$ to be the $[\Omega(26), \Omega(26)-26,3]_23,\ell$ code but now we consider $\ell=2$ with
$p(\Hb_0,2)=122$ that allows us, by \eqref{eq6:ell0=2inp}, to take $m\ge7$ and obtain the needed codes for $r\ge47$.

By \eqref{eq6:KnwFamR=3r=3t-1Ups}, \eqref{eq6:KnwSporadic}, \eqref{eq7:InfFamR=3New}, $\Delta(41,3)=\widehat{\Upsilon}(41)-\Omega(41)=(820\cdot2^5-2)-(818\cdot2^5-1)=63$, $\Delta(r,3)=\Upsilon(r)-\Omega(r)=2^{(r-26)/3}(819-818)=2^{(r-26)/3}$ if $r\ge44$.
 \end{proof}

In Tables \ref{tab7:LengFunrr33R=3} and \ref{tab7:LengFunr3464R=3}, the parameters $n,r$, and the covering density \eqref{eq1:CovDensity} of the best (as far as the authors know) binary linear covering $[n,n-r]_23$ codes of covering radius $R=3$ and redundancy (codimension) $3\le r\le64$  are written. The values of $n$ give an upper bound on the length function so that $\ell_2(r,3)\le n$; the asterisk * notes the exact value of $\ell_2(r,3)$. The new results, obtained in this paper in  Section~\ref{sec7:R=3new} are written in bold font and also are noted by the big star $\bigstar$. The sign $\emptyset$ in the column $p(\ell)$ notes that the code is not an $[n,n-r]_23,\ell$ code.

 \begin{table}[htbp]
\caption{The parameters of the best linear $[n,n-r]_23$ covering codes, $r\le33$,
such that $\ell_2(r,3)\le n$; the known values with the references seen in Section \ref{sec6:R=3known}; the new values in bold font and with $\bigstar$ seen in Section~\ref{sec7:R=3new}; $p(\ell)\ge p(\Hb,\ell)$; $\Delta=\Delta(r,3)$}
\centering
  \begin{tabular}
  {r|c|r|c|l|c|c|c|c}\hline
$r$&for $n$&\multicolumn{1}{c|}{$n$}&references&density$\le$&$p(0)$&$p(1)$&$p(2)$&$\Delta$\\\hline
3 &&$3^{*}$ &\cite{CHLL-bookCovCod,GrahSlo}&1&&&\\
4 &&$5^{*}$ &\cite{CHLL-bookCovCod,DDL-IEEE,GrahSlo}&1.62500 &3&5&\\
5 &&$6^{*}$ &\cite{CHLL-bookCovCod,GrahSlo}          &1.31250 &&&\\
6 &&$7^{*}$ &\cite{CHLL-bookCovCod,GrahSlo}          &1       &7&&\\
7 &&$11^{*}$&\cite{CHLL-bookCovCod,DDL-IEEE,GrahSlo}&1.81250&7&8& \\
8 &&$14^{*}$&\cite{CHLL-bookCovCod,GrahSlo}          &1.83594&10&& \\
9&&$18^{*}$&\cite{BaichVavLengthFunBin,BrPl1990,CHLL-bookCovCod,DMP-arXiv2025,DMP-IEEE2026,OstKaikUpBndBin1998
}&1.92969&$\bigstar\mathbf{9}$&11&$\bigstar\mathbf{\emptyset}$\\
10&&22      &\cite{GrahSlo}                          &1.75196&&&  \\
11&&$23^{*}$&\cite{Golay1949,GrahSlo}                &1      &23&&  \\ \hline
12&&37&\cite{CHLL-bookCovCod,LobsPleRevisTab1994,OstKaikUpBndBin1998}&2.06885&$\bigstar\mathbf{18}$&$
\bigstar\mathbf{\emptyset}$&$\bigstar\mathbf{\emptyset}$\\
13&&52      &\cite{KaikRoseADSlike2003}              &2.86609&&&\\
14&$w(14)$&63&\cite{CHLL-bookCovCod,Dav90PIT,DDL-ACCT2-1990,DDL-IEEE,DougJanCovRadCalc1991}&2.54688&&&\\
15&$\psi(15)$&71      &\cite{DavCovNewConstrCovCod}&1.82227&&&32\\
16&&95&\cite{CHLL-bookCovCod,DDL-IEEE}  &2.18165&&&\\
17&&126&\cite{DDL-IEEE}  &2.54442&&&\\
18&&\textbf{151}     &$\bigstar$  &\textbf{2.18946}&\textbf{19}&&&\textbf{2}\\
19&&205     &\cite{DDL-IEEE}  &2.73901&&&\\
20&&254&\cite{DavCovNewConstrCovCod}            &2.60486&&&\\
21&&303&\cite{DMP-arXiv2025,DMP-IEEE2026}  &2.21091&&35&&\\
22&$\varphi(22)$&383&\cite{CHLL-bookCovCod,DDL-IEEE}    &2.23255&&&\\
23&$w(23)$&511&\cite{CHLL-bookCovCod,Dav90PIT,DDL-ACCT2-1990,DDL-IEEE}  &2.65113&&&\\
24&&618     &\cite{DDL-IEEE}  &2.34477&&&\\
25&$\varphi(25)$&767&\cite{CHLL-bookCovCod,DDL-IEEE}  &2.24125&&&\\
26&$\mathbf{\Omega(26)}$&\textbf{817}&$\bigstar$&\textbf{1.35438}&\textbf{33}&\textbf{33}&\textbf{122}&$\mathbf{1}$\\
27&&1215    &\cite{DDL-IEEE}  &2.22726&&&&\\
28&$\varphi(28)$&1535&\cite{CHLL-bookCovCod,DDL-IEEE}    &2.24562&&&\\
29&$\upsilon(29)$&1642&\cite{CHLL-bookCovCod,DDL-ACCT2-1990,DDL-IEEE}  &1.37436&&&\\
30&$\psi(30)$&2303    &\cite{DavCovNewConstrCovCod}&1.89597&&&\\
31&$\varphi(31)$&3071&\cite{CHLL-bookCovCod,DDL-IEEE}    &2.24781&&&\\
32&$\upsilon(32)$&3286&\cite{CHLL-bookCovCod,DDL-ACCT2-1990,DDL-IEEE}  &1.37687&&&\\
33&$\psi(33)$&4607&\cite{DavCovNewConstrCovCod}&1.89721&&&\\\hline
 \end{tabular}
  \label{tab7:LengFunrr33R=3}
\end{table}

\begin{table}[htbp]
\caption{The parameters of the best linear $[n,n-r]_23$ covering codes, $34\le r\le64$,
such that $\ell_2(r,3)\le n$; the known values with the references (ref.) seen in Section \ref{sec6:R=3known}; the new values in bold font and with $\bigstar$ seen in Section~\ref{sec7:R=3new}; $p(\ell)\ge p(\Hb,\ell)$; $\Delta=\Delta(r,3)$}
\centering
  \begin{tabular}
  {r|c|r|c|l|c|c|c|c}\hline
$r$&for $n$&\multicolumn{1}{c|}{$n$}&ref.&density$\le$&$p(0)$&$p(1)$&$p(2)$&$\Delta$\\\hline
34&$\varphi(34)$&6143 &\cite{CHLL-bookCovCod,DDL-IEEE}&2.24891&&&\\
35&$\upsilon(35)$&6574 &\cite{CHLL-bookCovCod,DDL-ACCT2-1990,DDL-IEEE}&1.37813&&&\\
36&$\psi(36)$&9215 &\cite{DavCovNewConstrCovCod}            &1.89782&&&\\
37&$\varphi(37)$&12287&\cite{CHLL-bookCovCod,DDL-IEEE}&2.24946&&&\\
38&$\mathbf{\widehat{\Omega}(38)}$&\textbf{13102}&$\bigstar$&\textbf{1.36371}&\textbf{57}&\textbf{59}&&$\mathbf{2^4}$\\
39&$\psi(39)$&18431&\cite{DavCovNewConstrCovCod}            &1.89813&&&\\
40&$\varphi(40)$&24575&\cite{CHLL-bookCovCod,DDL-IEEE}&2.24973&&&\\
41&$\mathbf{\Omega(41)}$&\textbf{26175}&$\bigstar$&\textbf{1.35919}&\textbf{67}&\textbf{67}&&$\mathbf{63}$\\
42&$\psi(42)$&36863  &\cite{DavCovNewConstrCovCod}&1.89829&&&\\
43&$\varphi(43)$&49151	&\cite{CHLL-bookCovCod,DDL-IEEE}&2.24986&&&\\
44&$\mathbf{\Omega(44)}$&\textbf{52351}&$\bigstar$&\textbf{1.35927}&\textbf{131}&$\mathbf{131}$&&$\mathbf{2^6}$\\
45&$\psi(45)$&73727   &\cite{DavCovNewConstrCovCod}&1.89837&&& \\
46&$\varphi(46)$&98303&\cite{CHLL-bookCovCod,DDL-IEEE}&	2.24994&&&\\
47&$\mathbf{\Omega(47)}$   &\textbf{104703}&$\bigstar$&\textbf{1.35931}&&&\textbf{123}&$\mathbf{2^7}$\\
48&$\psi(48)$&147455   &\cite{DavCovNewConstrCovCod}&1.89840&&&\\
49&$\varphi(49)$&196607&\cite{CHLL-bookCovCod,DDL-IEEE}   &2.24997&&&\\
50&$\mathbf{\Omega(50)}$   &\textbf{209407}&$\bigstar$&\textbf{1.35933}&&&\textbf{123}&$\mathbf{2^8}$\\
51&$\psi(51)$&294911   &\cite{DavCovNewConstrCovCod}&1.89842&&&\\
52&$\varphi(52)$&393215&\cite{CHLL-bookCovCod,DDL-IEEE}   &2.24999&&&\\
53&$\mathbf{\Omega(53)}$   &\textbf{418815}&$\bigstar$&\textbf{1.35934}&&&\textbf{123}&$\mathbf{2^9}$\\
54&$\psi(54)$&589823   &\cite{DavCovNewConstrCovCod}&1.89843&&&\\
55&$\varphi(55)$&786431&\cite{CHLL-bookCovCod,DDL-IEEE}&2.25000&&&\\
56&$\mathbf{\Omega(56)}$& \textbf{837631}  &$\bigstar$&\textbf{1.35934}&&&\textbf{123}&$\mathbf{2^{10}}$\\
57&$\psi(57)$&1179647   &\cite{DavCovNewConstrCovCod}&1.89844&&&\\
58&$\varphi(58)$&1572863&\cite{CHLL-bookCovCod,DDL-IEEE}   &2.25000&&&\\
59&$\mathbf{\Omega(59)}$   &\textbf{1675263}&$\bigstar$&\textbf{1.35934}&&&\textbf{123}&$\mathbf{2^{11}}$\\
60&$\psi(60)$&2359295   &\cite{DavCovNewConstrCovCod}&1.89844&&&\\
61&$\varphi(61)$&3145727&\cite{CHLL-bookCovCod,DDL-IEEE}   &2.25000&&&\\
62&$\mathbf{\Omega(62)}$   &\textbf{3350527}&$\bigstar$&\textbf{1.35935}&&&\textbf{123}&$\mathbf{2^{12}}$   \\
63&$\psi(63)$&4718591   &\cite{DavCovNewConstrCovCod}&1.89844&&&   \\
64&$\varphi(64)$&6291455 &\cite{CHLL-bookCovCod,DDL-IEEE}  &2.25000&&&   \\ \hline
 \end{tabular}
  \label{tab7:LengFunr3464R=3}
\end{table}

\section*{Acknowledgments} The research of A. A. Davydov was carried out within the state assignment of Ministry of Science and Higher Education of the Russian Federation No.\ FFNU-2025-0028 for Kharkevich Institute for Information Transmission Problems. The research of S. Marcugini and F. Pambianco was supported in part by the Italian National Group for Algebraic and Geometric Structures and their Applications (GNSAGA--INDAM) (Contract No. U-UFMBAZ-2019-000160, 11.02.2019) and by University of Perugia
(Project No. 98751: Strutture Geometriche, Combinatoria e loro Applicazioni, Base Research
Fund 2017--2019; Fighting Cybercrime with OSINT, Research Fund 2021). This work was partially funded by the SERICS project (PE00000014) under the MUR National Recovery and Resilience Plan funded by the European Union-NextGenerationEU.

\end{document}